\documentclass[a4paper,11pt]{article}
\usepackage[utf8]{inputenc}
\usepackage[T1]{fontenc}
\usepackage[noadjust]{cite}
\usepackage{a4wide}
\usepackage{hyperref}
\hypersetup{colorlinks=true,citecolor=blue,linkcolor=blue,urlcolor=blue}
\usepackage{amsmath}
\usepackage{amssymb}
\usepackage{amsthm}   
\usepackage{enumerate}
\usepackage{thm-restate}
\usepackage[displaymath, mathlines,running]{lineno}

\usepackage{tikz}
\usetikzlibrary{fit,positioning, intersections, shapes.geometric,shapes.symbols,arrows,decorations.markings,decorations.pathreplacing,calc}

\theoremstyle{plain}
\newtheorem{thm}{Theorem}[section]
\newtheorem{cor}[thm]{Corollary}
\newtheorem{lem}[thm]{Lemma}
\newtheorem{obs}[thm]{Observation}
\newtheorem{clm}{Claim}[thm]

\theoremstyle{definition}
\newtheorem{defn}[thm]{Definition}
\newenvironment{claimproof}[1][\unskip]{\noindent {\emph{Proof of Claim #1.\space}}}{\hfill$\triangleleft$\smallskip}

\newcommand{\wei}{\mathfrak{w}}
\newcommand{\irr}{\operatorname{irr}}
\newcommand{\pfour}[1]{\mathsf{Adj}_{P_4}(#1)}
\newcommand{\bolda}{\mathbf{a}}

\newcommand{\boldF}{\mathbf{F}}

\newcommand{\boldR}{\mathbf{R}}
\newcommand{\boldS}{\mathbf{S}}
\newcommand{\boldb}{\mathbf{b}}
\newcommand{\boldu}{\mathbf{u}}
\newcommand{\boldv}{\mathbf{v}}

\newcommand{\boldx}{\mathbf{x}}
\newcommand{\boldy}{\mathbf{y}}
\newcommand{\boldf}{\mathbf{f}}
\newcommand{\boldg}{\mathbf{g}}
\newcommand{\N}{\mathbb{N}}
\newcommand{\bigO}{\mathcal{O}}

\newcommand{\calS}{\mathcal{S}}
\newcommand{\calT}{\mathcal{T}}
\newcommand{\calX}{\mathcal{X}}

\newcommand{\calP}{\mathcal{P}}

\newcommand{\calG}{\mathcal{G}}
\newcommand{\calH}{\mathcal{H}}
\newcommand{\calJ}{\mathcal{J}}
\newcommand{\calI}{\mathcal{I}}
\newcommand{\calR}{\mathcal{R}}
\newcommand{\calHplus}{\mathcal{H}^+}

\newcommand{\tw}{\mathrm{tw}}
\newcommand{\pw}{\mathrm{pw}}

\newcommand{\rom}[1]{%
	\textup{\uppercase\expandafter{\romannumeral#1}}%
}

\newcommand{\NEQ}{\mathrm{NEQ}}
\newcommand{\OR}{\mathrm{OR}}

\newcommand{\dist}{\mathrm{dist}}
\newcommand{\arity}{\mathrm{ar}}
\newcommand{\con}{;\!}

\newcommand{\sharpP}{\#$\mathrm{P}$}
\newcommand{\CSP}{\textsc{CSP}}
\newcommand{\Hom}[1]{\#\mathrm{Hom}(#1)}
\newcommand{\LHom}[1]{\#\mathrm{LHom}(#1)}
\renewcommand{\hom}[2]{\mathrm{hom}\bigl(#1 \to #2\bigr)}

\mathchardef\hyph="2D
\newcommand{\abs}[1]{|#1|}

\renewcommand{\epsilon}{\varepsilon}

\newcommand{\from}{\colon}

\usepackage{todonotes}

\newenvironment{myitemize}
{ \begin{itemize}
		\setlength{\itemsep}{0pt}
		\setlength{\parskip}{0pt}
		\setlength{\parsep}{0pt}     }
	{ \end{itemize}                  } 

\newenvironment{myenumerate}
{ \begin{enumerate}
		\setlength{\itemsep}{0pt}
		\setlength{\parskip}{0pt}
		\setlength{\parsep}{0pt}     }
	{ \end{enumerate} 				 }

\newcommand{\executeiffilenewer}[3]{%
\ifnum\pdfstrcmp{\pdffilemoddate{#1}}%
{\pdffilemoddate{#2}}>0%
{\immediate\write18{#3}}\fi%
} 

\newcommand{\svg}[2]{\def\svgwidth{#1}%
\executeiffilenewer{#2.svg}{#2.pdf}%
{inkscape -z -D --file=#2.svg %
--export-pdf=#2.pdf --export-latex}%
{\input{#2.pdf_tex}}}
\begin{document}
\title{Counting list homomorphisms from graphs of bounded treewidth: tight complexity bounds}

\author{
Jacob Focke\thanks{CISPA Helmholtz Center for Information Security. } \and 
D\'{a}niel Marx\thanks{CISPA Helmholtz Center for Information Security. Research supported by the European Research Council (ERC) consolidator grant No.~725978 SYSTEMATICGRAPH.} \and
Paweł Rzążewski\thanks{Warsaw University of Technology, Faculty of Mathematics and Information Science and University of Warsaw, Institute of Informatics, \texttt{p.rzazewski@mini.pw.edu.pl}. Supported by the Polish National Science Centre grant no. 2018/31/D/ST6/00062.}
}
\date{}

\begin{titlepage}
\def\thepage{}
\thispagestyle{empty}
\maketitle
\begin{abstract}
The goal of this work is to give precise bounds on the counting complexity of a family of generalized coloring problems (list homomorphisms) on bounded-treewidth graphs.
Given graphs $G$, $H$, and lists $L(v)\subseteq V(H)$ for every $v\in V(G)$, a {\em list homomorphism} is a function $f:V(G)\to V(H)$ that preserves the edges (i.e., $uv\in E(G)$ implies $f(u)f(v)\in E(H)$) and respects the lists (i.e., $f(v)\in L(v))$. 
Standard techniques show that if $G$ is given with a tree decomposition of width $t$, then the number of list homomorphisms can be counted in time $|V(H)|^t\cdot n^{\mathcal{O}(1)}$. Our main result is determining, for every fixed graph $H$, how much the base $|V(H)|$ in the running time can be improved. For a connected graph $H$ we define $\operatorname{irr}(H)$ in the following way: if $H$ has a loop or is nonbipartite, then $\operatorname{irr}(H)$ is the maximum size of a set $S\subseteq V(H)$ where any two vertices have different neighborhoods; if $H$ is bipartite, then $\operatorname{irr}(H)$ is the maximum size of such a set that is fully in one of the bipartition classes. For disconnected $H$, we define $\operatorname{irr}(H)$ as the maximum of $\operatorname{irr}(C)$ over every connected component $C$ of $H$.
It follows from earlier results that if $\operatorname{irr}(H)=1$, then the problem of counting list homomorphisms to $H$ is polynomial-time solvable, and otherwise it is \#P-hard.
We show that, for every fixed graph $H$,  the number of list homomorphisms from $(G,L)$ to $H$
\begin{itemize}
\item  can be counted in time $\operatorname{irr}(H)^t\cdot n^{\mathcal{O}(1)}$ if a tree decomposition of $G$ having width at most $t$ is given in the input, and
\item  given that $\operatorname{irr}(H)\ge 2$, cannot be counted in time $(\operatorname{irr}(H)-\epsilon)^t\cdot n^{\mathcal{O}(1)}$ for any $\epsilon>0$, even if a tree decomposition of $G$ having width at most $t$ is given in the input, unless the Counting Strong Exponential-Time Hypothesis (\#SETH) fails.
\end{itemize}
Thereby we give a precise and complete complexity classification featuring matching upper and lower bounds for all target graphs with or without loops.
\end{abstract}

\end{titlepage}
	
	\section{Introduction}
	Many of the NP-hard problems studied in the literature are known to be
polynomial-time solvable when restricted to graphs of bounded
treewidth. In fact, the majority of these problems can be solved in
time $f(t)\cdot n^{\bigO(1)}$ if a tree decomposition of width $t$ is
given in the input, that is, they are fixed-parameter tractable (FPT)
parameterized by treewidth. As algorithms working on tree
decompositions are useful building blocks in many types of FPT results
and approximation schemes, determining the optimal 
dependence $f(t)$ on the width of the decomposition received significant attention.

On the upper bound side, new algorithmic techniques (such as Fast Subset Convolution, Cut \& Count, representative sets) were developed to obtain improved algorithms. For lower bounds, conditional complexity results were given ruling out certain forms of running times. Lokshtanov, Marx, and Saurabh~\cite{DBLP:journals/talg/LokshtanovMS18} considered problems that are known to be solvable in time $c^t\cdot n^{\bigO(1)}$ if a tree decomposition of width $t$ is given in the input, and showed that these algorithms are essentially optimal, as no algorithm with running time $(c-\epsilon)^t\cdot n^{\bigO(1)}$ can exist for any $\epsilon>0$, assuming the Strong Exponential-Time Hypothesis (SETH). In particular, for \textsc{Vertex Coloring} with $c$ colors, the textbook $c^t\cdot n^{\bigO(1)}$ algorithm based on dynamic programming cannot be improved to $(c-\epsilon)^t\cdot n^{\bigO(1)}$. By now, there is a growing collection of tight lower bounds of this form in the literature \cite{DBLP:journals/siamcomp/OkrasaR21,DBLP:conf/esa/OkrasaPR20,DBLP:conf/stacs/EgriMR18,DBLP:conf/soda/CurticapeanLN18,
DBLP:conf/iwpec/BorradaileL16,DBLP:journals/dam/KatsikarelisLP19, DBLP:conf/icalp/MarxSS21, DBLP:conf/soda/CurticapeanM16}. 

Vertex coloring with $c$ colors can be seen as a homomorphism problem. Given graphs $G$ and $H$, a {\em homomorphism} is a mapping $f:V(G)\to V(H)$ that preserves the edges of $G$, that is, $uv\in E(G)$ implies $f(u)f(v)\in E(H)$. Let us observe that a $c$-coloring of $G$ can be seen equivalently as a homomorphism from $G$ to $K_c$, the complete graph with $c$ vertices. Thus generalized coloring problems defined by homomorphisms to a fixed graph $H$ were intensively studied in the combinatorics literature and subsequently from the viewpoint of computational complexity \cite{DBLP:journals/jct/HellN90,DBLP:journals/csr/HellN21,DBLP:journals/csr/HellN08,DBLP:books/daglib/0013017,DBLP:journals/dam/MaurerSW81,DBLP:journals/iandc/MaurerSW81b,Hahn1997}. The homomorphism problem can be generalized from graphs to arbitrary relational structures, giving a very direct connection to constraint satisfaction problems (CSPs), which are often described using the terminology of homomorphisms \cite{DBLP:journals/csr/HellN21,DBLP:books/daglib/0013017,DBLP:journals/csr/HellN08,DBLP:journals/siamcomp/FederV98}.

Introducing lists of allowed images gives a more robust variant for homomorphism problems: formally, given graphs $G$ and $H$, a {\em list assignment} is a function $L:V(G)\to 2^{V(H)}$. Then a list homomorphism from $(G,L)$ to $H$ is a homomorphism $f:V(G)\to V(H)$ that additionally respects the lists, that is, $f(v)\in L(v)$ for every $v\in V(G)$. Studying the list version of a homomorphism problem can be seen as analogous to studying the conservative version of CSP, where every unary constraint is allowed \cite{DBLP:journals/jcss/Bulatov16,DBLP:journals/tocl/Bulatov11,DBLP:conf/lics/Barto11}. The study of conservative CSP often served as a starting point before more general investigations, motivating the exploration of the list version of homomorphism problems.

The polynomial-time solvable cases of (list) homomorphism is well understood. Ne\v{s}et\v{r}il and Hell \cite{DBLP:journals/jct/HellN90} showed that the non-list problem is polynomial-time solvable if $H$ is bipartite or has a loop, and NP-hard when restricted to any other $H$. The complexity of the list version was characterized by Feder, Hell, and Huang \cite{DBLP:journals/jgt/FederHH03}: it is polynomial-time solvable if $H$ is a so-called bi-arc graph, and NP-hard for every other fixed $H$. But are there nontrivial algorithmic ideas that can help us obtain improved running times for the NP-hard cases? 
Similarly to coloring, the problem of finding a (list) homomorphism to $H$ can be solved in time $|V(H)|^t\cdot n^{\bigO(1)}$ if $G$ is given with a tree decomposition of width $t$. Can this straightforward dynamic programming algorithm be improved for a given fixed $H$, and if so, by how much?

This question was resolved first by Egri, Marx, and  Rz\k{a}\.zewski~\cite{DBLP:conf/stacs/EgriMR18} for the list homomorphism problem in the special case when $H$ is reflexive (that is, every vertex has a loop), which was extended by Okrasa, Piecyk, and  Rz\k{a}\.zewski~\cite{DBLP:conf/esa/OkrasaPR20} to every $H$ (where each vertex may or may not have a loop). It turns out that there are a couple of algorithmic ideas that can be used to obtain $c^t\cdot n^{\bigO(1)}$ time algorithms with $c<|V(H)|$. In particular, there are delicate notions of decompositions of $H$, such that if they are present, then the problem can be reduced in a nontrivial way to multiple instances with some $H'$ having fewer vertices than $H$. To formalize the optimality of these ideas, a combinatorial parameter $i^*(H)$ was introduced, and it was shown that the problem can be solved in time $i^*(H)^t\cdot n^{\bigO(1)}$, but there is no $(i^*(H)-\epsilon)^t\cdot n^{\bigO(1)}$ time algorithm for any $\epsilon>0$, assuming the SETH. This means that the technical ideas behind the $i^*(H)^t\cdot n^{\bigO(1)}$ time algorithm already capture all the possible algorithmic insights that can be exploited when solving the problem on a given tree decomposition. A similar tight result was obtained by Okrasa and  Rz\k{a}\.zewski~\cite{DBLP:journals/siamcomp/OkrasaR21} for the non-list version of the problem. In that version, a different set of algorithmic ideas become relevant (reduction to a homomorphic core and factorizations), and the optimality of these ideas were proved assuming not only the SETH, but also two long-standing conjectures from algebraic graph theory.

A well-known phenomenon in computational complexity is that even if it is possible to find a solution efficiently in a combinatorial problem, counting the number of solutions can be hard. The most notable example is the perfect matching problem in bipartite graphs: finding a perfect matching is polynomial-time solvable, but counting the number of perfect matchings is \#P-hard by the seminal result of Valiant~\cite{DBLP:journals/siamcomp/Valiant79}. 
There are algorithmic ideas that can be generalized from decision to counting (for example, simple forms of dynamic programming), but others may not be. In particular, arguments of the form ``if there is a solution, then there is a solution with property $P$, hence we only need to look for solutions with property $P$'' do not immediately generalize to counting, as they would fail to count the potential solutions that do not have property $P$.
Unfortunately, the $i^*(H)^t\cdot n^{\bigO(1)}$ time algorithm of Okrasa, Piecyk,  and  Rz\k{a}\.zewski~\cite{DBLP:conf/esa/OkrasaPR20} for list homomorphism heavily relies on such arguments. Careful observation shows that only a very limited set of algorithmic ideas remain relevant for the counting problem:
\begin{itemize}
\item \textbf{Connected components.} We may assume that $G$ is connected: otherwise the number of solutions is the product of the number of solutions for each component of $G$. Furthermore, if $G$ is connected, then every vertex of $G$ is mapped into the same connected component of $H$. Thus we may assume that $H$ is connected as well: the number of solutions is the sum of the number of solutions when restricted to each component of $H$.
\item \textbf{Identical neighborhoods.} If two vertices $u$ and $v$ have identical neighborhoods, then we may pretend that only one of them, say $u$, is present in $H$. Then we need to solve a weighted problem that takes into account that whenever a vertex of $G$ is mapped to $u$, then we can actually choose from two different copies of $u$. The extension to this weighted problem can be easily implemented in the dynamic programming algorithm working over a tree decomposition. Therefore, the number of possible images of a vertex can be restricted to the maximum size of a set $S\subseteq V(G)$ where the neighborhoods are pairwise different. We call such a set $S$ an {\em irredundant} set.
  \item \textbf{Bipartite classes.} If $H$ is bipartite (and has no loops), then $G$ has to be bipartite as well. This means that each bipartite class of $G$ is mapped to one of the two bipartite classes of $H$, resulting in two cases that we can treat separately. In each case, the possible images of a vertex $v$ of $G$ are restricted to a bipartite class of $H$. Now the maximum number of possible images of $v$ can be bounded by the maximum size of an irredundant set $S$ that is  contained in one bipartite class of $H$. 
  \end{itemize}
  Our main negative result shows that these are all the algorithmic ideas that can be exploited in the list homomorphism problem for a fixed $H$. 
  To state this formally, we define the following graph parameter.
  \begin{defn}\label{def:irr}
    Given a graph $H$, we say that $S\subseteq V(H)$ is {\em irredundant} if any two vertices have different neighborhoods. We say that graph $H$ is irredundant if $V(H)$ is irredundant. For a connected bipartite graph $H$, set $S$ is {\em one-sided} if it is fully contained in one of the bipartite classes.
    
For a connected graph $H$, we define $\irr(H)$ the following way. If $H$ has a loop or is nonbipartite, then $\irr(H)$ is the maximum size of an irredundant set; if $H$ is bipartite, then $\irr(H)$ is the maximum size of a one-sided irredundant set. For disconnected $H$, we define $\irr(H)$ as the maximum of $\irr(C)$ over every connected component $C$ of $H$.
\end{defn}
Note that if $v$ has a loop, then the neighborhood of $v$ includes $v$ itself. This has to be carefully taken into account when interpreting different neighborhoods. For example, in a reflexive clique (where every vertex has a loop) every vertex has the same neighborhood.

For a fixed graph $H$ (possibly with loops), we denote by $\LHom{H}$ the problem of counting the number of list homomorphisms from the given $(G,L)$ to $H$. 
Our main result shows for $\LHom{H}$ that $\irr(H)$ is indeed the best possible base of the exponent that can appear in the running time. We obtain the lower bound under the \#SETH, the counting version of the SETH, which is the natural variant of the assumption for lower bounds on counting problems (see \cite{DellHusfeldtMarxTaslamanWahlen}). Note that the SETH implies the \#SETH, thus proving a result under the latter assumption make it somewhat stronger.
%
%\begin{thm}
%	Let $H$ be a graph that has a component which is not a biclique nor a reflexive clique.
%	On $n$-vertex instances with pathwidth $\pw$, $\LHom{H}$
%	\begin{enumerate}
%		\item can be solved in time $\irr(H)^\pw \cdot \poly(n)$, provided that a path decomposition of width $\pw$ is given as part of the input.
%		\item cannot be solved in time $(\irr(H)-\epsilon)^\tw \cdot \poly(n)$, for any $\epsilon >0$, even if a path decomposition of width $\pw$ is given as part of the input, unless the \#SETH fails.
%	\end{enumerate}
%\end{thm}

\begin{thm} \label{thm:mainthm}
	Let $H$ be a graph that has a component which is not a biclique nor a reflexive clique.
	On $n$-vertex instances with treewidth $\tw$, the $\LHom{H}$ problem
	\begin{myenumerate}
		\item can be solved in time $\irr(H)^\tw \cdot n^{\bigO(1)}$, provided that a tree decomposition of width $\tw$ is given as part of the input.
		\item cannot be solved in time $(\irr(H)-\epsilon)^\tw \cdot n^{\bigO(1)}$, for any $\epsilon >0$, even if a tree decomposition of width $\tw$ is given as part of the input, unless the \#SETH fails.
	\end{myenumerate}
\end{thm}

While the counting complexity of homomorphism problems is well-studied, we may argue that our lower bound brings a new level of understanding on the real hardness of the problem. It is easy to deduce from earlier results that if $H$ has a component that is neither a biclique nor a reflexive clique, then $H$ contains an induced subgraph $H'$ on at most 4 vertices such that $\LHom{H'}$ is \#P-hard. As we can use the lists to restrict the problem to $V(H')$, it follows that $\LHom{H}$ is \#P-hard as well. Dyer and Greenhill \cite{DyerG00:RSA} showed that for such an $H$ the homomorphism counting problem is \#P-hard even if no lists are allowed. Since we cannot use the lists to restrict the problem to $V(H')$, this proof is significantly more complicated. What all these proofs have in common is that they identify only one level of complexity: \#P-hardness. Intuitively, counting 10-colorings feels harder than counting 3-colorings, as we need to consider a larger number of possible values at each vertex. Saying that both problems are \#P-hard ignores this perceived difference of hardness. On the other hand, Theorem~\ref{thm:mainthm} establishes, in a formal way, a difference in complexity between $\LHom{K_{10}}$ and $\LHom{K_3}$: the best possible base is 10 in the former and 3 in the latter.

Typical \#P-hardness proofs use as the basis of the reduction a few known \#P-hard problems, such as counting independent sets or counting 3-colorings. Once it is established that the problem can express one of these hard problems, then the proof can stop, as it is irrelevant whether some even harder problem can also be expressed. On the other hand, the lower bound proof of Theorem~\ref{thm:mainthm} has to delve deeper into the complexity of the problem, as it should extract as much hardness from the problem as possible. So instead of trying to show that the problem can express a simple relation corresponding to e.g., independent set or 3-coloring, our main goal is to show that, in a formal sense, $\LHom{H}$ can express every relation over a domain of size $\irr(H)$. Going beyond the context of bounded-treewidth graphs, the main message of this paper is that the $\LHom{H}$ problem is sufficiently rich to express all such relations.

\subsection{Proof overview}
As described above, the algorithmic part of the main result is fairly simple: we need to do only a few changes to the standard dynamic programming algorithm to take into account connected components, vertices with identical neighborhoods, and bipartiteness. Therefore, we focus on the main steps of the lower bound proof in this section.

\paragraph{From SAT to CSP.} The \#SETH states that for every $\epsilon>0$, there is a $d\ge 1$ such that $n$-variable \#\textsc{$d$-SAT} (i.e., counting satisfying assignments of a Boolean formula where every clause has $d$ literals) cannot be solved in time $\bigO((2-\epsilon)^n)$. Therefore, we need to show that the hypothetical fast algorithm for $\LHom{H}$ would give an algorithm for \#\textsc{SAT} violating this lower bound. It will be convenient to introduce an intermediate problem. By \#\CSP($q$, $r$), we denote the problem of counting solutions in a CSP instance  with domain $[q]$ and each constraint having arity at most $r$. It follows from the work of Lampis~\cite{DBLP:journals/siamdm/Lampis20} that if for some $q\ge 2$ there is an $\epsilon>0$ such that \#\CSP($q$, $r$) can be solved in time $(q-\epsilon)^n\cdot (n+m)^{\bigO(1)}$ for every $r$, then the \#SETH fails. Therefore, to prove our main lower bound for $\LHom{H}$, we take a  \#\CSP($q$, $r$) instance with $q=\irr(H)$ and show that it can be reduced to $\LHom{H}$.

\paragraph{\boldmath Focusing on bipartite $H$.} Similarly to \cite{DBLP:conf/esa/OkrasaPR20}, first we prove the main result for bipartite $H$. Then it is a surprisingly simple task to raise the result to general $H$. Given a (not necessarily bipartite) graph $H$, we define the {\em associated} bipartite graph $H^*$ the following way (see Figure~\ref{fig:assoc}): for every vertex $v$ of $H$, there are are two vertices $v',v''$ in $H^*$ and we introduce edges such that, if $uv\in E(H)$, then $u'v'', u''v'\in E(H^*)$ (if $v$ has a loop in $H$, then $v'v''\in E(H^*)$). It can be shown that $\irr(H^*)=\irr(H)$ and the lower bound for $\LHom{H}$ can be obtained by a simple reduction from $\LHom{H^*}$. Thus in the following we can assume that $H$ is bipartite. Furthermore, we may assume that $H$ is irredundant, that is, any two vertices have distinct neighborhoods. It is easy to see that the lower bound for irredundant $H$ can be extended to the general case. 
\begin{figure}[t]
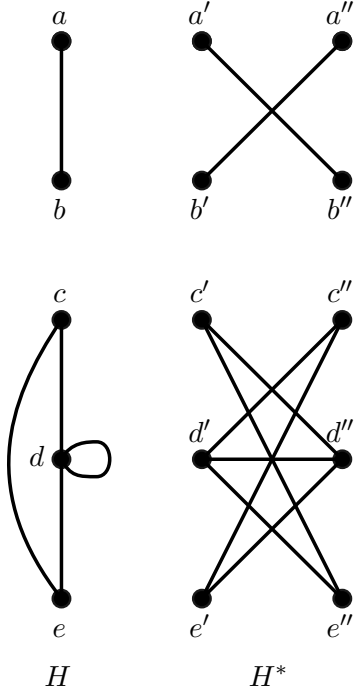

\begin{center}
\svg{0.3\linewidth}{assoc}
\end{center}
\caption{A graph $H$ and its associated bipartite graph $H^*$.}
\label{fig:assoc}
\end{figure}
  
\paragraph{New relations.} Input to $\LHom{H}$ can be seen as a CSP instance where the vertices are variables with domain $V(H)$ and each edge is a binary constraint restricting the combination of values appearing on two vertices. We can augment the problem by allowing a fixed finite set $\mathcal{R}=\{R_1,\dots, R_c\}$ of other relations in the instance. Of course, introducing new relations may make the problem harder. If for every $1 \le i \le c$, we can provide an appropriate reduction from the problem augmented with $\{R_1,\dots, R_i\}$ to the problem augmented only with $\{R_1,\dots, R_{i-1}\}$, then this chain of reductions shows that the problem augmented with $\mathcal{R}$ is not harder than the original $\LHom{H}$ problem, hence any lower bound proved for the augmented problem holds for $\LHom{H}$ as well. Therefore, we start with $\LHom{H}$ and try to introduce new relations $R_i$ one by one, in a way that does not make the problem any harder. We do this until we can introduce every $r$-ary relation over a domain $S$ of size $q=\irr(H)$. At this point, the lower bound for \#\CSP($q$, $r$) applies to our augmented problem, and hence to the original $\LHom{H}$ problem as well.

\paragraph{Gadgets and interpolation.} How can we show that introducing a new relation $R_i$ does not make the problem harder? This can be done by showing that each occurrence of the $R_i$ relation can be replaced by an appropriate gadget that, in a certain sense, simulates this relation. A {\em gadget} is a small instance where a subset $(x_1,\dots, x_r)$ of the vertices are distinguished as the {\em interface.} Now if we fix the images $f(x_1)=v_1$, $\dots$, $f(x_r)=v_r$, then, for each vector $(v_1,\dots,v_r)$, there is some number of ways the mapping $f$ can be extended to a list homomorphism from the gadget to $H$. If it so happens that for every $(v_1,\dots, v_r)\in R_i$ there is exactly one extension, and for every $(v_1,\dots, v_r)\not\in R_i$ there are zero extensions, then this gadget effectively expresses a constraint with the relation $R_i$ on the vertices $(x_1,\dots, x_r)$. Now we can replace every occurrence of $R_i$ with a copy of this gadget, showing that introducing $R_i$ does not make the problem harder.

However, it is rare that a gadget can express the new relation $R_i$ in such a clean way. Fortunately, the nature of counting problems allows us to use gadgets in a much more general setting. Suppose the number of possible extensions for every $(v_1,\dots, v_r)\in R_i$ is always an integer from a set $A$ (for example, $A=\{4,5,8,11\}$), while the number of extensions for $(v_1,\dots, v_r)\not\in R_i$ is always an integer from a set $B$ (for example, $B=\{3,7,9\}$). Suppose that every $a\in A$ and $b\in B$ are coprime\footnote{More precisely, we allow 0 to appear only in $B$, and allow 1 to appear only in $A$.} (that is, there is no prime $p$ dividing both an element of $A$ and an element of $B$); this condition holds in our example. Then a polynomial interpolation technique of Dyer and Greenhill~\cite{DyerG00:RSA} allows us to count those assignments where $(v_1,\dots, v_r)\in R_i$ is satisfied, effectively introducing a constraint with relation $R_i$.

A simple application of this technique allows us to express compositions of relations. Given two binary relations $R_1$ and $R_2$, their {\em composition}
is defined as $R_1\con R_2=\{(a,b) \mid \exists\,c\colon (a,c)\in R_1, (c,b)\in R_2\}$. Given two interface vertices $v_1$ and $v_2$, we can construct a gadget by introducing a new vertex $z$ and adding the constraint $R_1$ on $v_1$ and $z$, and adding the constraint $R_2$ on $z$ and $v_2$. Now if the interface vertices are assigned a pair of values from $R_1\con R_2$, then there is at least one extension to $z$; if they are assigned a pair of values not from $R_1\con R_2$, then there is no extension. Thus by the technique mentioned in the previous paragraph, we can express the relation $R_1\con R_2$.

\paragraph{Relations on a $P_4$.} The path on four vertices (a $P_4$) is the smallest example of a connected bipartite graph $H$ that is not a complete bipartite graph, and hence the  $\LHom{H}$ problem is \#P-hard. Intuitively, \#P-hardness should mean that whenever a $P_4$ appears in $H$, then we should be able to express arbitrarily complicated constraints when restricted to these vertices; however, this has not been stated explicitely in earlier work. We make this expectation formal by showing that if $(a-b-c-d)$ is a $P_4$ appearing in $H$, then any relation $R\subseteq \{a,c\}^r$ can be simulated with appropriate gadgets.  For this statement it is essential that we consider the counting problem, and an analogous statement should not hold for the decision version. To see this, observe that in the path $(a-b-c-d)$, the neighborhood of $c$ is a superset of the neighborhood of $a$. Thus if the list of a vertex contains both $a$ and $c$, then $a$ in the solution can always be replaced by $c$. This makes it impossible to simulate even very simple relations such as $R=\{aa,cc\}$ in the decision version since allowing $aa$ and $cc$ would automatically allow $ac$  and $ca$ as well.

\paragraph{Forcers.}
Forcers are crucial objects in our proof. Let $S\subseteq V(H)$, $x,y\in S$, and $\alpha,\beta\in V(H)$, where $S$, $\{x,y\}$, $\{\alpha,\beta\}$ are all one-sided sets. An $(x,y,S)$-forcer with respect to $(\alpha,\beta)$ is a binary relation $R\subseteq S\times \{\alpha,\beta\}$ that
\begin{itemize}
\item contains $(x,\alpha)$, but not $(x,\beta)$,
\item contains $(y,\beta)$, but not $(y,\alpha)$, and
\item for every $z\in S\setminus \{x,y\}$, contains at least one of $(z,\alpha)$ or $(z,\beta)$.
\end{itemize}
Intuitively, the relation moves $x$ and $y$ to $\alpha$ and $\beta$, respectively, while other values of $S$ can move to somewhere in $\{\alpha,\beta\}$.
For example, both of the following two relations are $(x,y,\{x,y,z,q,w\})$-forcers with respect to $(\alpha,\beta)$:
\begin{align*}
  R_1&=\{x\alpha,y\beta, z\alpha,z\beta,q\alpha,w\beta\},\\
  R_2&=\{x\alpha,y\beta, z\alpha,q\alpha,w\alpha\}.
\end{align*}

We will be particularly interested in forcers where $\alpha$ and $\beta$ are vertices $a$ and $c$ of a $P_4$ $a-b-c-d$. The following claim will be crucial for our proof: if we can show that, for every $x,y\in S$, there are good $(x,y,S)$-forcers with respect to $(a,c)$ of the same $P_4$, then we can realize any $r$-ary relation $R$ over $S$. Therefore, if we choose an $S$ with $|S|=\irr(H)=q$, then we can
reduce \#\CSP($q$, $r$) to our problem, which was our main goal. Let us sketch how to prove the claim by realizing an $r$-ary relation $R$ with the help of the forcers.
Let $s=|S|$. Intuitively, by applying all the $s(s-1)$ possible $(x,y,S)$-forcers on a vertex $v$, we can extract $s(s-1)$ bits of information about the value of $v$. These bits are sufficient to identify the value of $v$, as for any two different values $x\neq y$, one of the bits will be different. Therefore, we can translate the values of $r$ variables into $rs(s-1)$ bits and then we can implement $R$ by an $rs(s-1)$-ary relation over $\{a,c\}$. As $a$ and $c$ are on a $P_4$, we have already seen that such a relation can be realized.

\paragraph{The connected structure of $P_4$'s.} 
As we have seen in the previous paragraph, a crucial step in the proof is to realize an $(x,y,S)$-forcer for every $x,y\in S$. An important detail is that we need all these $(x,y,S)$-forcers to be with respect to the same $(a,c)$. But what happens if we have two $P_4$s $(a-b-c-d)$ and $(a'-b'-c'-d')$, can an $(x,y,S)$-forcer with respect to $(a',c')$ be turned into one with respect to $(a,c)$? We can show that if the two $P_4$s are ``adjacent'' in the sense that they share two vertices in the same bipartite class of $H$, then this is indeed possible. Even if the two $P_4$s are not adjacent, but are in the ``same connected component'' under this definition of adjacency, then multiple applications of this argument show that a forcer with respect to $(a',c')$ can be turned into a forcer with respect to $(a,c)$.  A graph-theoretic analysis shows that if $H$ is connected and irredundant, then the $P_4$s have only a single connected component. Therefore, it does not matter in the definition of good forcers exactly where $a$ and $c$ are, as long as they are in a $P_4$. We can show that even this last condition is unnecessary: a forcer with respect to any two distinct $(\alpha,\beta)$ that are on the same side can be turned into a forcer with respect to $(a,c)$. 

\paragraph{The inductive proof.}
We prove by induction that there is a good $(x,y,S)$-forcer for every $S\subseteq V(H)$ and $x,y\in S$: when proving the statement for $(x,y,S)$, we assume that it is true for every $(x',y',S')$ where either
\begin{itemize}
\item $|S'|<|S|$ or
\item $|S'|=|S|$ and the distance condition $\dist(\{x',y'\},S'\setminus \{x',y'\})<\dist(\{x,y\},S\setminus \{x,y\})$ holds.
\end{itemize}

To explain the main ideas of the induction step, we need two new concepts. 
A weaker version of the forcer is the distinguisher, which additionally may allow $(y,\alpha)\in R$ as well. For example, the following relation is an $(x,y,\{x,y,z,q,w\})$-distinguisher with respect to $(\alpha,\beta)$:
\begin{align*}
  R&=\{x\alpha,y\alpha,y\beta, z\alpha,z\beta,q\alpha,w\beta\}.
\end{align*}
As another application of the interpolation technique, we show that if we can realize a distinguisher, then we can realize a forcer as well, that is, we can eliminate the unwanted possibility $(y,\alpha)$.

For a partition $(X,Y)$ of $S$, the $(X,Y)$-partitioner with respect to $(\alpha, \beta)$ is the relation
\[
R=\{(x,\alpha)\mid x\in X\} \cup \{(y,\beta)\mid y\in Y\}.
\]
It can be seen as a stricter version of the $(x,y,S)$-forcer: it is now specified for every element of $S$ whether the relation should move it to $\alpha$ or $\beta$. But we can show that if we realize an $(x,y,S)$-forcer  for every $x,y\in S$, then we can realize an $(X,Y)$-partitioner for every partition $(X,Y)$ of $S$: the argument is similar to how forcers were used to realize arbitrary relations over $S$.

Now let us overview the main cases in the inductive proof.

\begin{itemize}
\item \textbf{Base case:} $S=\{x,y\}$. If $x$ and $y$ have a common neighbor, then  $x$ and $y$ are part of a $P_4$ (as we assumed that $H$ is irredundant) and we essentially have an $(x,y,S)$-distinguisher, which we can turn into an $(x,y,S)$-forcer. If $x$ and $y$ have distance at least 3, then we can use the edge relation $E$ of $H$ to move $x$ and $y$ closer to each other, until they have a common neighbor.

\begin{figure}
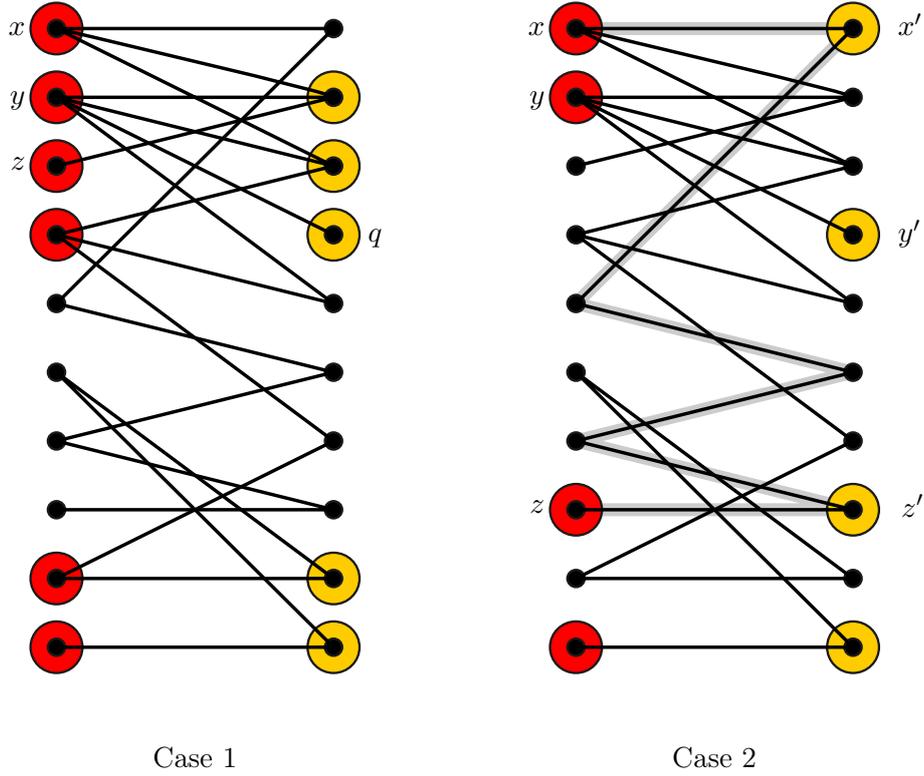

  \begin{center}
    \svg{0.8\linewidth}{cases}
    \caption{The two cases in the inductive proof. The set $S$ is highlighted in red, the set $S'$ is highlighted in yellow. Case 1: As $z$ and $y$ have a common neighbor, we can define $S'$ with $|S'|<|S|$. Vertex $q\in S'$ is in $N(y)\setminus N(x)$. Case 2: In $S$, the distance of $\{x,y\}$ from the rest of the vertices is $6$ (a shortest path $P$ is highlighted in gray), while in $(x',y',S')$ this distance is only $4$.}
    \label{fig:cases}
  \end{center}
\end{figure}

\item \textbf{Case 1:} $S$ contains a vertex $z$ at distance two from $\{x,y\}$. Using that $z$ has a neighbor that is also adjacent to one of $x$ and $y$, we can define a set $S'$ with $S\subseteq N(S')$ and $|S'|<|S|$ (see Figure~\ref{fig:cases}). Furthermore, we can define $S'$ in a way that, without loss of generality, the set $N(y)\setminus N(x)$ is nonempty. As $|S'|<|S|$, the induction hypothesis implies that there are $(x',y',S')$-forcers for every $x',y'\in S'$. Therefore, we may also assume that there are $(X',Y')$-partitioners for every partition $(X',Y')$ of $S'$.
Let $R$ be an $(N(x)\cap S',S'\setminus N(x))$-partitioner with respect to some $(\alpha,\beta)$. If $E$ is the edge relation of $H$, then we claim that the composition $E\con R$ is an $(x,y,S)$-distinguisher. Indeed, it moves $x$ to $\alpha$ (as $E$ has to move $x$ to $N(x)\cap S'$ first, which is moved to $\alpha$ by $R$), $y$ can move to $\beta$ (as $E$ can move $y$ to the nonempty $(N(y)\setminus N(x))\cap S'$, which is moved to $\beta$ by $R$), and every element of $S$ can move to at least one of $\alpha$ and $\beta$ (as $S\subseteq N(S')$). Then this $(x,y,S)$-distinguisher can be turned into an $(x,y,S)$-forcer, as required.

\item \textbf{Case 2:} Every vertex of $S\setminus \{x,y\}$ is at distance at least 4 from $\{x,y\}$. Let $P$ be a shortest path between $\{x,y\}$ and $S\setminus \{x,y\}$. There are a few similar cases to consider, but let us assume, as a representative case, that $P$ goes from $x$ to some vertex $z$, and there is a vertex $y'\in N(y)\setminus N(x)$  (see Figure~\ref{fig:cases}). Let $x'$ be the neighbor of $x$ on $P$, let $z$ be the other endpoint of $P$, and let $z'$ be the neighbor of $z$ on $P$. We define $S'$ starting from $\{x',y',z'\}$ and adding a neighbor of each $S\setminus \{x,y,z\}$; clearly, we have $|S'|\le |S|$. Furthermore, $\dist(\{x',y'\},S'\setminus \{x',y'\})$ is strictly less than $\dist(\{x,y\}, S\setminus \{x,y\})$: vertices $x'$ and $z'$ are closer to each other than $x$ and $z$. Thus by the induction hypothesis, an $(x',y',S')$-forcer $R$ exists with respect to some $(\alpha,\beta)$. Then we claim that the composition $E\con R$ is an $(x,y,S)$-distinguisher. Indeed, $x$ is moved to $\alpha$ (as $x'$ is the only neighbor of $x$ in $S'$, which is moved to $\alpha$ by $R$), $y$ can move to $\beta$ (as $y'$ is a neighbor of $y$ and $R$ can move $y'$ to $\beta$), and every element of $S$ can move to at least one of $\alpha$ and $\beta$ (as $S\subseteq N(S')$). Then this $(x,y,S)$-distinguisher can be turned into an $(x,y,S)$-forcer, as required.
\end{itemize}

%%% Local Variables:
%%% mode: latex
%%% TeX-master: "CountLHomTW"
%%% End:

\section{Preliminaries} \label{sec:preliminaries}
For an integer $n$, we define $[n] := \{1,2,\ldots,n\}$.
Given a set $S$, by $2^S$ we denote the power set of $S$.

\paragraph{Graph theory}
By $(p_1-p_2-\ldots -p_q)$ we denote the path of length $q$ whose consecutive vertices are $p_1, p_2,\ldots, p_q$. 
Let $H$ be a graph.
For vertices $u, v\in V(H)$, $\dist(u,v)$ is the length of a shortest path between $u$ and $v$ in $H$. $N(u)=\{v\in V(H) \mid uv \in E(H)\}$ is the (open) neighborhood of $u$ in $H$.  For a set $U\subseteq V(H)$, $N(U)$ is the set of vertices of $H$ with at least one neighbor in $U$.

\paragraph{Structures}
	A signature $\sigma$ consists of a finite set of relation symbols with specified arities. For a relation (symbol) $R$, $\arity(R)$ denotes the arity of $R$. For $i\in [\arity(R)]$, $\Pi_i(R)$ is the projection onto the $i$-th entry of the tuples in $R$.
	A \emph{structure} $\calH$ with signature $\sigma(\calH)$ consists of a universe $V(\calH)$ together with a set of relations $\boldR(\calH)=\{R^{\calH} \mid R\in \sigma(\calH)\}$ over the universe $V(\calH)$. We write $\calH=(V(\calH), \boldR(\calH))$.
	By $||\calH||$ we denote $|\sigma(\calH)| + |V(\calH)| + \sum_{R\in \sigma(\calH)} |R^{\calH}|\cdot\arity(R)$, which is the size of a ``reasonable'' encoding of $\calH$.
	Note that if $\sigma(\calH)$ consists of a constant number of relations, each with constant arity, then $||\calH|| = |V(\calH)|^{\bigO(1)}$. Throughout the paper we work in such a setting.

	A structure $\calH$ is a graph if $\boldR(\calH)=\{E(\calH)\}$, where $E(\calH)$ is a symmetric binary relation over $V(\calH)$. In this case we write $\calH=(V(\calH),E(\calH))$.
	
	The \emph{Gaifman graph} of a structure $\calH$ has vertex set $V(\calH)$ and distinct $u,v \in V(\calH)$ are adjacent if and only if there is a  relation $R \in \boldR(\calH)$ that involves $u$ and $v$.
	We define the treewidth and a tree decomposition (resp., the  pathwidth and a path decomposition) of $\calH$ as the treewidth and a tree decomposition (resp., the pathwidth and a path decomposition) of the Gaifman graph of $\calH$.
	
	Given two structures $\calG$ and $\calH$ with the same signature $\sigma$, a function $f\from V(\calG) \to V(\calH)$ \emph{respects} $R\in \sigma$ if, for each $\boldx\in R^{\calG}$, $f(\boldx)\in R^{\calH}$, where $f$ is evaluated elementwise.  A \emph{homomorphism} from $\calG$ to $\calH$ is a function $h\from V(\calG) \to V(\calH)$ that respects every $R\in \sigma$. By $\hom{\calG}{\calH}$ we denote the number of homomorphisms from $\calG$ to $\calH$.
	If $\arity(R)=1$ then $R^\calH$ is a subset of $V(\calH)$ and we also use the usual \emph{list} notation: For an element $x\in V(\calG)$ we write $L(x)= R^\calH$ if $x\in R^\calG$. More generally, if $L$ is a partial function from $V(\calG)$ to $2^{V(\calH)}$, then $h$ is a homomorphism from $(\calG,L)$ to $\calH$ if $h$ is a homomorphism from $\calG$ to $\calH$ with $h(x)\in L(x)$ for each $x$ in the domain of $L$. Note that $x$ being outside the domain of $L$ is as restrictive as having $L(x)=V(\calH)$. By $\hom{(\calG, L)}{\calH}$ we denote the number of homomorphisms from $(\calG,L)$ to $\calH$.

	Let  $\mathcal{H} = (U, \boldS)$ be a fixed structure
	and let $\mathcal{J} = (Z, \boldF)$ be a structure with the same signature as $\mathcal{H}$.
	For some integer $p$, let $\boldx = (x_1,x_2,\ldots,x_p) \in Z^p$ and $\boldy = (y_1,y_2,\ldots,y_p) \in U^p$.
	Then $\hom{(\calJ,\boldx)}{(\calH,\boldy)}$ is the number of homomorphisms from $\calJ$ to $\calH$ that map $x_i$ to $y_i$ for each $i\in[p]$.
	A $p$-tuple of distinguished (not necessarily pairwise distinct) elements of $Z$ is also called an \emph{interface} (of size $p$) of $\calJ$.
	
	\paragraph{Counting Problems}
	Let $\calH=(V,\boldR)$ be a structure. Then $\Hom{\calH}$ is the problem that takes as input a structure $\calG$ with the same signature as $\calH$ and asks to determine $\hom{\calG}{\calH}$.
	
	Note that the power set $2^V$ is the set of all unary relations over $V$. Then $\LHom{\calH}$ is the problem $\Hom{\calHplus}$, where $\calHplus=(V,\boldR\cup (2^V\setminus \{\emptyset\}))$. Intuitively, $\calHplus$ is the structure obtained from $\calH$ by adding all non-empty unary relations over the universe of $\calH$.
	In the special case, where $\calH$ is a graph $H$ we simplify the notation by writing instances of $\LHom{H}$ as $(G,L)$, where $G$ is a graph and $L$ is a function from the vertices of $G$ to subsets of $V(H)$ specifying the unary relations of the input (lists). For convenience we assume that for such instances every vertex of $G$ has a list. Then $\LHom{H}$ on input $(G,L)$ asks to determine $\hom{(G,L)}{H}$.

\paragraph{Pathwidth and pathwidth-preserving reductions}	
Let $A$ and $B$ be computational problems that take as input some structure given along with a path decomposition.
A \emph{pathwidth-preserving reduction} from $A$ to $B$ is a polynomial-time Turing reduction from $A$ to $B$ such that the corresponding algorithm, if executed on an instance $\calG$ of $A$ given with a path decomposition with width at most $t$,
makes $B$-oracle calls only on structures of size polynomial in $||\calG||$, given with a path decomposition with width at most $t + \bigO(1)$.%

Let $\calH=(V, \boldR)$ be a structure. A relation $R\subseteq V^{\arity(R)}$ is \emph{realizable} (by $\calH$) if there is a pathwidth-preserving reduction from $\LHom{(V, \boldR\cup \{R\})}$ to $\LHom{\calH}$.

Note that pathwidth-preserving reductions are transitive in the sense that a pathwidth-preserving reduction from $A$ to $B$ together with a pathwidth-preserving reduction from $B$ to $C$ implies that there is a pathwidth-preserving reduction from $A$ to $C$.
The concept of realizable relations is transitive as well: If a relation $R$ is realizable by a structure $\calH=(V,\boldR)$, and a relation $R'$ is realizable by $\calH'=(V, \boldR\cup\{R\})$, then $R'$ is realizable by $\calH$.

However, there is a caveat here: we are only allowed to combine a constant number of pathwidth-preserving reductions to make sure that the total increase of the pathwidth bound is constant as well.

We say that $\calH = (V, \boldR)$ is an \emph{induced substructure} of $\calH'=(V',\boldR')$ if $V \subseteq V'$
and for every $R$ we have $\boldx \in R^{\calH}$ if and only if $\boldx \in R^{\calH'} \cap V^{\arity(R)}$.
Note that, since we can use arbitrary lists, we can see each instance of $\LHom{\calH}$ as an instance of $\LHom{\calH'}$, where each element of the input universe has a list that is a subset of $V$. This gives us a simple pathwidth-preserving reduction from $\LHom{\calH}$ to $\LHom{\calH'}$.
Therefore, if a relation $R$ is realizable by $\calH$, it is also realizable by $\calH'$.
We will use this fact implicitly throughout this work.

\paragraph{Relations and Operators}
	
Let us now define some notation concerning relations.
\begin{itemize}
\item Let $R_1 \subseteq X_1 \times X_2$ and $R_2\subseteq X_2 \times X_3$. The \emph{composition} of $R_1$ and $R_2$ is the relation $R_1\con R_2 \subseteq X_1 \times X_3$ defined as $\{(u,v) \mid \exists\,w\colon (u,w)\in R_1, (w,v)\in R_2\}$.
\item Given a graph $H$ and two sets $U,W\subseteq V(H)$, we define $R_{U\to W} =\{(u,v) \mid u\in U, v\in W\cap N(u)\}$.  
\item For $q \geq 2$, a relation $R \subseteq X_1 \times \cdots \times X_q$, and $v_1 \in X_1$, we define $R(v_1)=\{ (v_2,\ldots,v_q) \in X_2 \times \cdots \times X_{q} \mid (v_1,v_2,\ldots,v_q) \in R\}$.
\end{itemize}
Note that for a graph $H$, sets $U,W \subseteq V(H)$, and $u \in U$, we have $R_{U \to W}(u) = \{v\in W \mid (u,v) \in R_{U\to W}\}=V\cap N(u)$.

The following lemma shows three simple ways to obtain realizable relations.

\begin{restatable}{lem}{lemsimple}
%\begin{lem}
\label{lem:simplecombinations}
Let $\calH=(V, \boldR)$. The following relations $R$ are realizable by $\calH$:
\begin{myenumerate}
\item $R_1 \cap R_2$, where $R_1, R_2 \in \boldR$ are of the same arity,
\item $R_{U \to U'}$, where $\boldR$ contains $E(H)$ for some graph $H$ with vertex set $V$, and $U,U' \subseteq V$,
\item $R_1 \con R_2$, where $R_1, R_2 \in \boldR$ are of arity 2.
\end{myenumerate}
%\end{lem}	
\end{restatable}

The proof of Lemma~\ref{lem:simplecombinations} is postponed to the next subsection, where we introduce necessary tools to prove the third item.

\subsection{Pathwidth-Preserving Reductions and Interpolation}
One well-known tool for  proving hardness for exact counting problems is interpolation. 

\begin{lem}[{\cite[Lemma~3.2]{DyerG00:RSA}}]\label{lem:interpolation}
	For distinct non-zero integers $\bolda=(a_1, \dots, a_k)$ and integers $\boldb=(b_1, \dots b_k)$, suppose that, for each $j\in [k]$, we have $\sum_{i=1}^k a_i^j\cdot x_i=b_j$ for some unknown $\boldx=(x_1, \dots x_k)$. Then the values of $\boldx$ can be computed in time polynomial in $k$.
\end{lem}

The following lemma is a modification of~\cite[Corollary 3.3]{DyerG00:RSA} that generalizes to structures and makes the reduction pathwidth-preserving. Note that the conditions on the numbers are slightly different than the conditions on the weights in~\cite[Corollary 3.3]{DyerG00:RSA}.

\begin{lem}\label{lem:newRelation}\label{lem:primes}
Let $\calH =(V,\boldR)$ be a structure with signature $\sigma$ and, for some positive integer $p$, let $R\subseteq V^{p}$ be a relation that is not in $\boldR$.
Suppose that there is a structure $\calJ$ with signature $\sigma$ and an interface $\boldx$ of size $p$ such that, for each $\boldf\in R$ and $\boldg \in V^p \setminus R$, we have the following:
\begin{myenumerate}
	\item $\hom{(\calJ,\boldx)}{(\calH,\boldf)} \neq 0$.\label{prop:positive}
	\item $\hom{(\calJ,\boldx)}{(\calH,\boldg)}\neq 1$.\label{prop:no1}
	\item If $\hom{(\calJ,\boldx)}{(\calH,\boldg)}\neq0$, then there is no prime that divides both $\hom{(\calJ,\boldx)}{(\calH,\boldf)}$ and $\hom{(\calJ,\boldx)}{(\calH,\boldg)}$. \label{prop:primefactor}
\end{myenumerate}
Then $R$ is realizable by $\calH$. 
\end{lem}
	
\begin{proof}
%Suppose we have an algorithm that can solve every instance of $\LHom{H}$ in time $q^{\tw} \cdot \poly$, given a tree decomposition of the Gaifman graph of the input structure with width $\tw$.
In order to show that $R$ is realizable by $\calH$, we need to show, for  $\calH'=(V,\boldR\cup\{R\})$, that there is a pathwidth-preserving reduction from $\LHom{\calH'}$ to $\LHom{\calH}$. 

Consider an input $\calG = (U, \boldS)$ of the $\LHom{\calH'}$ problem, given along with a path decomposition with width $t$.
Slightly overloading notation, we will use $R$ to denote both the relation and its symbol so the corresponding relation in $\calG$ is $R^{\calG}$.
We set $n=||\calG||$ and $m = |R^\calG|$.
Note that if $m=0$, then we are done, as $\calG$ can be cast as a structure with signature $\sigma$ and it therefore can be interpreted as an instance of the  $\LHom{\calH}$ problem.

Let $k$ be a positive integer.	
For each $(v_1,v_2,\ldots,v_p) \in R^\calG$,
we introduce $k$ copies of the structure ${\calJ}$ with interface $\boldx=(x_1, \ldots, x_p)$.
For each copy we identify $x_1,\ldots, x_p$ with $v_1,v_2,\ldots,v_p$, respectively.
Then we remove the relation $R^\calG$ from $\calG$.
Let us call the obtained structure $\calG_{\calJ^k}$.
Note that the signature of $\calG_{\calJ^k}$ is $\sigma$.

Let $\Phi$ be the class of functions from $U$ to $V$ that respect all relations from $\sigma$ (but not necessarily $R$).
Let $\Phi^+$ be the class of functions in $\Phi$ that respect $R$.

For each function $f \in \Phi$, define 
\[w(f) = \prod_{\boldv\in R^\calG} \hom{(\calJ,\boldx)}{(\calH,f(\boldv))}.\]
We also say that $f$ is a $w(f)$-function. For each integer $w$, by $\Phi(w)$ we denote the set of $w$-functions in $\Phi$.

Define
\begin{align*}
\mathcal{W} := & \{ w(f) \mid f \in \Phi \text{ and } w(f) \neq 0\}\\
\mathcal{W}^+ := & \{ w(f) \mid f \in \Phi^+  \text{ and } w(f) \neq 0\}\\
\mathcal{W}^- := & \{ w(f) \mid f \in \Phi\setminus \Phi^+  \text{ and } w(f) \neq 0\}.
\end{align*}

First note that for $f\in \Phi^+$, we have $w(f)\neq 0$ according to property~\ref{prop:positive} as given in the statement of the lemma. Thus, $\mathcal{W}^+ = \{ w(f) \mid f \in \Phi^+\}$.

We now argue that $\mathcal{W}^+$ and $\mathcal{W}^-$ are disjoint.
Let $w^+\in \mathcal{W}^+$ and $w^-\in \mathcal{W}^-$.
We have $w^-\neq 0$, and by properties~\ref{prop:no1} and~\ref{prop:primefactor} as given in the statement of the lemma, we observe that $w^-$ has a prime factor that does not divide $w^+$. Therefore, $w^+\neq w^-$. 
We can conclude that
\begin{equation}\label{equ:LHomAnswer}
	\hom{\mathcal{G}}{\mathcal{H'}} = \sum_{w\in \mathcal{W}^+} |\Phi(w)|.
\end{equation}
Furthermore, we have
\begin{align}
	\hom{\mathcal{G}_{\mathcal{J}^k}}{\calH}
	&= \sum_{f \in \Phi} w(f)^k \nonumber\\
	&= \sum_{w\in \mathcal{W}} \sum_{f\in \Phi(w)} w^k  \nonumber\\
	&= \sum_{w\in \mathcal{W}} w^k\cdot |\Phi(w)|. \label{equ:wFcts}
\end{align}
	
Let $\bolda=(w)_{w\in \mathcal{W}}$ and $\boldb=(\hom{\mathcal{G}_{\mathcal{J}^k}}{\calH})_{k\in [\abs{\mathcal{W}}]}$.
Note that a factor of the form $\hom{(\calJ,\boldx)}{(\calH,f(\boldv))}$ can have at most $|V|^p$ different values, each of which can have multiplicity between $0$ and $m$ in the product $w(f)$. Thus, there are at most $(m+1)^{(\abs{V}^p)}$ distinct values for $w(f)$ and the set $\mathcal{W}$ can be computed in time polynomial in $n$.
		
So, the tuple $\bolda$ can be computed in time polynomial in $n$. The tuple $\boldb$ can be computed using the algorithm that solves $\LHom{\calH}$ for each $k\in [\abs{\mathcal{W}}]$ on input $\mathcal{G}_{\mathcal{J}^k}$
We have $\abs{V(\mathcal{G}_{\mathcal{J}^k})}\le \abs{U} + \abs{R^\calG}\cdot \abs{V(\calJ)}\cdot k$.

Consider a tuple $(v_1,v_2,\ldots,v_p) \in R^\calG$ and note that it corresponds to a clique in the Gaifman graph of $\calG$,
so in a path decomposition $\calX$ of this graph there is a bag that contains $\{v_1,v_2,\ldots,v_p\}$.
Let $X'$ be the first such a bag. We modify $\calX$ by inserting right after $X'$ bags $X_1',X_2',\ldots,X_k'$, where $X_i'$ is the union of $X'$ and the vertex set of the $i$-th copy of $\calJ$ with interface $(v_1,v_2,\ldots,v_p)$.
Clearly the obtained sequence $\calX'$ is a path decomposition of $\calG_{\calJ^k}$. As the size of $\calJ$ depends only on $R$ and $\calH$, and thus is a constant, we obtain that the width of $\calX'$ is $t + \bigO(1)$.
	
Thus, $\boldb$ can be computed using $|\mathcal{W}|$ calls to the $\LHom{\calH}$ oracle, and we have shown that $|\mathcal{W}|\in n^{\bigO(1)}$ and each of the oracle calls is on an instance given with a path decomposition of width at most $t + \bigO(1)$.
		 	
Note that the values of $\bolda$ are non-zero by definition of $\mathcal{W}$.
Using~\eqref{equ:wFcts} and the fact that $\bolda$ contains non-zero pairwise distinct values, we can apply Lemma~\ref{lem:interpolation}. Thus, we obtain $\boldx=(|\Phi(w)|)_{w\in \mathcal{W}}$ from $\bolda$ and $\boldb$ in time polynomial in $n$. By computing, for each $\boldf \in R$, the value of $\hom{(\calJ,\boldx)}{(\calH,\boldf)}$, one can decide in polynomial time which values of $\mathcal{W}$ belong to $\mathcal{W}^+$ (recall that each $w^-\in \mathcal{W}^-$ has a prime factor that does not appear in any value of $\hom{(\calJ,\boldx)}{(\calH,\boldf)}$ with $\boldf \in R$). Thus, using~\eqref{equ:LHomAnswer}, the sought-for number of (list) homomorphisms can be computed from $\boldx=(|\Phi(w)|)_{w\in \mathcal{W}}$ in time polynomial in $n$.
\end{proof}%	

The following simple corollary formalizes the modeling of relations by gadgets.
\begin{cor}\label{cor:simpleGadget}
	Let $\calH =(V,\boldR)$ be a structure with signature $\sigma$ and, for some positive integer $p$, let $R\in V^{p}$ be a relation whose relation symbol is not in $\sigma$.
	Suppose that there is a structure $\calJ$ with signature $\sigma$ and an interface $\boldx$ of size $p$ such that, for each $\boldf\in R$ and $\boldg \in V^p \setminus R$, we have $\hom{(\calJ,\boldx)}{(\calH,\boldf)} \neq 0$ and $\hom{(\calJ,\boldx)}{(\calH,\boldg)}=0$, then $R$ is realizable by $\calH$. 
\end{cor}

Now we are ready to prove Lemma~\ref{lem:simplecombinations} that we recall here.
\lemsimple*
\begin{proof}
For each case we will build a structure $\calJ=(U, \boldS)$ with the same signature as $\calH$ and interface $\boldx$, which satisfies the assumptions of Corollary~\ref{cor:simpleGadget}.
\begin{enumerate}
\item Let $R = R_1 \cap R_2$ for some $R_1,R_2 \in \boldR$ with $p=\arity(R_1)=\arity(R_2)$. We set $U = \{x_1, \ldots, x_p\}$, $\boldx=(x_1, \ldots, x_p)$, and $R_1^{\calJ} = R_2^{\calJ}=\{(x_1, \ldots, x_p)\}$.
\item Suppose $\boldR$ contains $E(H)$ for some graph $H$ with vertex set $V$, and $R = R_{U \to U'}$, where $U,U' \subseteq V$. We set $U = \{x,y\}$, $\boldx=(x,y)$, and $E(H)^{\calJ}=\{(x,y), (y,x)\}$. We also introduce lists $L(x) = U$, $L(y)=U'$.
\item Let $R = R_1 \con R_2$ for some  $R_1, R_2 \in \boldR$ or arity 2.
We set $U = \{x,y,z\}$, $\boldx=(x,z)$, $R_1^{\calJ} = \{(x,y)\}$, and $R_2^{\calJ}=\{(y,z)\}$.
\end{enumerate}
Note that in the first two cases all vertices of $\calJ$ are in $\boldx$, so it it straightforward to observe that 
$\hom{(\calJ,\boldx)}{(\calH,\boldf)} = 1$ if $\boldx \in R$ and $\hom{(\calJ,\boldx)}{(\calH,\boldf)} = 0$ if $\boldf \notin R$.
In the last case, we observe that $\hom{(\calJ,\boldx)}{(\calH,\boldf)} > 0$ if $\boldf \in R$ and $\hom{(\calJ,\boldx)}{(\calH,\boldf)} = 0$ otherwise.
Thus the lemma follows from Corollary~\ref{cor:simpleGadget}.
\end{proof}

\subsection{Hardness of counting satisfying assignments to  \CSP($q$, $r$)}	
Recall that the \#SETH states that for every $\epsilon>0$, there is a $d\ge 1$ such that $n$-variable \#\textsc{$d$-SAT}  cannot be solved in time $\bigO((2-\epsilon)^n)$.
Note that since $d$ is a constant, this is equivalent to saying that $n$-variable $m$-clause \#\textsc{$d$-SAT}  cannot be solved in time $(2-\epsilon)^n \cdot (n+m)^{\bigO(1)}$.

For integers $r,q \geq 2$, by \CSP($q$, $r$) we denote the CSP problem with domain $[q]$ and each constraint of arity at most $r$.
By \#\CSP($q$, $r$) we denote the problem of counting satisfying valuations of a given instance of \CSP($q$, $r$).
In this section we show the hardness of computing \#\CSP($q$, $r$), conditioned on the \#SETH.

The decision version of this result, conditioned on the SETH, was proven by Lampis~\cite{DBLP:journals/siamdm/Lampis20}. Our proof is just an adaptation of the original one to the counting world, so we will not elaborate on the details and refer the reader to the paper of Lampis~\cite{DBLP:journals/siamdm/Lampis20}.

\begin{thm}\label{thm:qcsp}
For every integer $q \geq 2$ and $\epsilon >0$ there is an integer $r$, such that the following holds.
Unless the \#SETH fails, \#\CSP($q$, $r$) with $n$ variables and $m$ constraints cannot be solved in time $(q-\epsilon)^n \cdot (n+m)^{\bigO(1)}$.
\end{thm}

\begin{proof}
Fix $q \geq 2$ and $\epsilon > 0$ and suppose that for every $r$ there is an algorithm solving $n$-variable $m$-constraint \#\CSP($q$, $r$) in time $(q-\epsilon)^n \cdot (n + m)^{\bigO(1)}$.
Suppose we are given a \#$d$-\textsc{SAT} instance $\Phi$ with $N$ variables and $M$ clauses, where $d$ is a constant.
Without loss of generality we may assume that each variable is involved in some clause.

In the reduction we will carefully choose $r = r(d,q,\epsilon)$ and build a \#\CSP($q$, $r$) instance $\Psi$.
To avoid confusion, we will refer to an assignment of values to the variables of $\Phi$ as \emph{truth assignment},
while an assignment of values to the variables of $\Psi$ will be called \emph{valuation}.

First, select an integer $p$ and a real $\delta > 0$, such that there exists an integer $t$ satisfying the following (see~\cite{DBLP:journals/siamdm/Lampis20} for the argument that such a choice is possible):
\[
(q-\epsilon)^p \leq (2-\delta)^t \leq 2^t \leq q^p.
\]
Note that $\delta$ does not depend on $d$.
Now group the variables of $\Phi$ into $\gamma := \lceil N/t \rceil$ groups, each with at most $t$ variables.
Call these groups $V_1,V_2\ldots,V_\gamma$.
For each $i \in [\gamma]$, we create a group $X_i$ of $p$ variables of $\Psi$.
Thus the total number of variables in $\Psi$ is $n:= \gamma \cdot p \leq \frac{pN}{t} + p$.

Note that the total number of truth assignments of variables in $V_i$ is at most $2^t \leq q^p$, so it does not exceed the number of possible valuations of $X_i$. Let us fix some injective function that maps each truth assignment on $V_i$ to a distinct valuation of $X_i$.

Now, let us define the constraints of $\Psi$.
Consider any clause $c$ of $\Phi$ and let $V_{i_1},V_{i_2},\ldots,V_{i_{s'}}$ for $d' \leq d$ be the groups that contain the variables of $c$.
Note that each truth assignment of variables in $V_{i_1},V_{i_2},\ldots,V_{i_{d'}}$ that satisfies $c$ corresponds to some valuation of variables in $X_{i_1},X_{i_2},\ldots,X_{i_{d'}}$.
We introduce a new constraint $C(c)$ on all variables from $X_{i_1} \cup X_{i_2} \cup \ldots \cup X_{i_{d'}}$ that accepts only the valuations that correspond to the truth assignments that satisfy $c$.
The arity of $C(c)$ is  $|X_{i_1} \cup X_{i_2} \cup \ldots \cup X_{i_{d'}}| \leq d \cdot p$.

Summing up, $\Psi$ has  $n= \gamma \cdot p \leq \frac{pN}{t} + p$ variables, $m = M$ constraints, its domain is $[q]$, and each constraint has arity at most $r := dp$.

It is straightforward to observe that $\Phi$ has a satisfying truth assignment if and only if $\Psi$ has a satisfying valuation.
However, a stronger property holds as well: there is a bijection between truth assignments that satisfy $\Phi$ and valuations that satisfy $\Psi$.
Indeed, recall that every variable of $\Phi$ appears in some clause, and thus, for each $i\in[\gamma]$, there is a constraint of $\Psi$ involving all variables of $X_i$. Since the constraints of $\Psi$ were defined in a way that the only accepted valuations are in one-to-one correspondence to the truth assignments of $\Phi$, we observe that the claimed bijection between truth assignments and valuations exists.
Hence, the solution to \#\CSP($q$, $r$) on the instance $\Psi$ is precisely the number of satisfying assignments of $\Phi$, i.e., the solution of our instance of \#$d$-\textsc{SAT}.
Let us call our hypothetical algorithm for $\Psi$. Its running time is:
\begin{align*}
(q - \epsilon)^n \cdot (n+m)^{\bigO(1)} \leq & \left((q - \epsilon)^p\right)^{N/t+1} \cdot (N + M)^{\bigO(1)} \\
 \leq & \left((2 - \delta)^t\right)^{N/t+1} \cdot (N+M)^{\bigO(1)} \\ \leq & 2^t \cdot (2-\delta)^{N} \cdot (N+M)^{\bigO(1)}.
\end{align*}
As $t$ depends only on $q$ and $\epsilon$, this running time is $(2-\delta)^{N} \cdot (N+M)^{\bigO(1)}$, which violates the \#SETH.
\end{proof}

\section{Counting list homomorphisms to bipartite graphs $H$}	\label{sec:bipartite}
Let $H$ be a bipartite graph. We say that a set $S \subseteq V(H)$ is \emph{irredundant} if for all distinct $u,v \in S$ it holds that $N(u) \neq N(v)$. We say that a graph $H$ is \emph{irredundant} if $V(H)$ is irredundant. If $H$ is connected, then $S \subseteq V(H)$ is \emph{one-sided} if it is contained in one bipartition class.
Then recall from Definition~\ref{def:irr} that $\irr(H)$ is the maximum size of a one-sided irredundant subset of $V(H)$.
We extend this to disconnected bipartite graphs by setting $\irr(H)$ to be the maximum value of $\irr(H')$ over all connected components $H'$ of $H$.
Observe that the following conditions are equivalent for every bipartite graph $H$:
\begin{enumerate}[(i)]
\item $\irr(H) \geq 2$,
\item $H$ has a connected component that is not a biclique,
\item $H$ contains an induced $P_4$.
\end{enumerate}

For a connected bipartite graph $H$, an instance $(G,L)$ of $\LHom{H}$ is \emph{consistent}, if:
\begin{itemize}
\item $G$ is connected and bipartite, let $X$ and  $Y$ be its bipartition classes,
\item $\bigcup_{x \in X} L(x)$ is contained in one bipartition class of $H$,
and $\bigcup_{y \in Y} L(y)$ is contained in the other bipartition class of $H$.
\end{itemize}

\subsection{Algorithm}
\begin{thm}\label{thm:bipartitealgo}
For each bipartite graph $H$, the $\LHom{H}$ problem on $n$-vertex instances given along with a tree decomposition of width at most $t$
can be solved in time $\irr(H)^t \cdot n^{\bigO(1)}$.
\end{thm}
\begin{proof}
Let $(G,L)$ be an instance of $\LHom{H}$, where $G$ has $n$ vertices and is given along with a tree decomposition of width at most $t$.
First observe that if $G$ is not bipartite, then there is no homomorphism from $G$ to $H$, thus we return 0.
So from now on assume that $G$ is bipartite.

First, assume that $G$ and $H$ are both connected.
Let $X,Y$ be the bipartition of $G$ and $A,B$ be the bipartition of $H$.
We observe that in every homomorphism from $G$ to $H$,
either each vertex of $X$ is mapped to a vertex of $A$ and each vertex of $Y$ is mapped to a vertex of $B$, or
each vertex of $X$ is mapped to a vertex of $B$ and each vertex of $Y$ is mapped to a vertex of $A$.

Thus we can reduce the problem to solving two consistent instances  of $\LHom{H}$ as follows.
Let $L_1$ be the lists obtained from $L$ by setting $L_1(x) := L(x) \cap A$ for every $x \in X$ and  $L_1(y) := L(y) \cap B$ for every $y \in Y$.
Similarly, define $L_2$ as follows: $L_2(x) := L(x) \cap B$ for every $x \in X$ and  $L_2(y) := L(y) \cap A$ for every $y \in Y$.
By the above reasoning, we obtain that $\hom{(G,L)}{H} = \hom{(G,L_1)}{H} + \hom{(G,L_2)}{H}$.

So let us focus on computing $\hom{(G,L_1)}{H}$ as computing $\hom{(G,L_2)}{H}$ is analogous.
We define lists $L_1'$ and an auxiliary weight function $\wei \from V(G) \times V(H) \to \N$ as follows.
For each $v \in V(G)$, we partition the vertices of $L_1(v)$ into subsets consisting of vertices with exactly the same neighborhood.
From each such subset $W$ we include in $L'_1(v)$ exactly one vertex, say $u$, and set $\wei(v,u) = |W|$.
On the other hand, for every $u \notin L'_1(v)$ we set $\wei(v,u)=0$.

Observe that for each $v \in V(G)$ the set $L'_1(v)$ is irredundant and contained in one bipartition class of $H$.
Thus, for each $v \in V(G)$, we have $|L'_1(v)| \leq \irr(H)$.

Furthermore, denoting by $\Omega$ the set of all homomorphisms from $(G,L'_1)$ to $H$, we observe that:
\begin{equation}
\hom{(G,L_1)}{H} = \sum_{f \in \Omega} \prod_{v \in V(G)} \wei(v,f(v)).
\label{eq:count-irredundant}
\end{equation}

Since every list in $L_1'$ has size at most $\irr(H)$, a standard bottom-up dynamic programming approach can be used to compute $\hom{(G,L'_1)}{H}$ in time $\irr(H)^t \cdot n^{\bigO(1)}$~\cite{DBLP:journals/tcs/DiazST02}.
Furthermore one can easily modify the algorithm to actually determine the sum in~\eqref{eq:count-irredundant}:
whenever we decide to assign color $u$ to some vertex $v \in V(G)$, we multiply the number of solutions by $\wei(w,u)$.
This modification does not increase the running time.

Now consider the general case that graphs $G$ and $H$ are possibly disconnected.
Let $G_1,G_2,\ldots,G_p$ be the connected components of $G$ and let $H_1,H_2,\ldots,H_q$ be the connected components of $H$.
It is straightforward to observe that
\[
\hom{(G,L)}{H} = \prod_{i = 1}^p \left( \sum_{j=1}^q \hom{(G_i,L)}{H_j}\right).
\]
Thus, the total running time of such an algorithm is $\irr(H)^t \cdot n^{\bigO(1)}$.
\end{proof}

\subsection{Hardness for bipartite target graphs}
The following lemma is the main building block of our hardness reduction.

\label{lem:allrelations}
\begin{restatable}{lem}{allrelations} 
	%\begin{lem} \label{lem:allrelations}
		Let $H=(V,E)$ be a connected bipartite graph with $\irr(H)\ge 2$, and let $S\subseteq V$ be a one-sided irredundant set.
		For every fixed $p \geq 1$, every relation $R \subseteq S^p$ is realizable by $H$.
	%\end{lem}
\end{restatable}

We postpone the proof of Lemma~\ref{lem:allrelations} to Section~\ref{sec:allrelations},
and now we show how it implies the lower bounds for bipartite graphs $H$.

\begin{thm}
\label{thm:bipartitehardness}
Let $H$ be a bipartite graph with $\irr(H)\ge 2$.
Assuming the \#SETH, there is no $\epsilon >0$, such that $\LHom{H}$ on consistent $n$-vertex instances given with a path decomposition of width $t$ can be solved in time $(\irr(H)-\epsilon)^t \cdot n^{\bigO(1)}$.
\end{thm}
\begin{proof}
Let $H$ be as in the assumptions and let $S$ be a maximum-size irredundant set contained in a bipartition class of some connected component $H'$ of $H$. Let $q := |S| = \irr(H)$. Note that $H'$ and $S$ satisfy the assumptions of Lemma~\ref{lem:allrelations}.
Suppose that there is $\epsilon>0$ and an algorithm that solves every consistent $n$-vertex $m$-edge instance $(G,L)$ of $\LHom{H}$ that is given along with a path decomposition of $G$ with width $t$ in time $(q-\epsilon)^t \cdot (n+m)^{\bigO(1)}$.

We reduce from \#\CSP($q$, $r$), where $r$ is the constant given for $q$ and $\epsilon$ by Theorem~\ref{thm:qcsp}.
Consider an instance $\Psi$ with variables $U$, domain $D = [q]$, and let $\calR$ be the set of relations used in the constraints of $\Psi$. Note that $|\calR|$ depends only on $q$ and $r$, i.e., $|\calR|$ is a constant. Furthermore, the number of constraints in $\Psi$ is polynomial in $|U|$.

Recall that $|S|=q$, so by fixing an arbitrary bijection between $S$ and $[q]$, we can equivalently see $\Psi$ as an instance with domain $S$.
In other words, the instance $\Psi$ can be equivalently seen as a structure, which is an instance of $\LHom{(S, \calR)}$.
Note that the pathwidth of $\Psi$ is clearly at most $|U|$.

As the arity of each relation in $\calR$ is at most $r$, by Lemma~\ref{lem:allrelations}, all relations in $\calR$ are realizable by $H'$ and thus by $H$. This means that there is a pathwidth-preserving reduction from $\LHom{(S, \calR)}$ to $\LHom{H}$.

Let us point out that our reduction is in fact a chain of pathwidth-preserving reductions. However, the total number of steps in this chain is $\bigO(|\calR|)$, which is a constant. Thus the total increase of the pathwidth is $\bigO(1)$.

So, using our hypothetical algorithm for $\LHom{H}$, we can count satisfying assignments for $\Psi$ in time
\[
(q-\epsilon)^{|U| + \bigO(1)} \cdot |U|^{\bigO(1)} = (q-\epsilon)^{|U|} \cdot |U|^{\bigO(1)},
\]
where we use the fact that $q$ is a constant.
By Theorem~\ref{thm:qcsp} the existence of such an algorithm for \#\CSP($q$, $r$) contradicts the \#SETH.
\end{proof}

\subsection{Special case: $H = P_4$}
As a warm-up, let us start with the special case that $H= P_4$.
The following lemma can be seen as a restriction of Lemma~\ref{lem:allrelations} to this case.
The proof will illustrate our approach.

\begin{lem}\label{lem:basicP4}
Consider $P_4 =(a-b-c-d)$ and fix a positive integer $q$. Then any $R \subseteq \{a,c\}^q$ is realizable by $P_4$.
\end{lem}
\begin{proof}
We first show that the relations 
$\NEQ :=  \{(a,c), (c,a)\}$ and $\OR_q :=  \{a,c\}^q \setminus \{c^q\}$
are realizable by $P_4$. We will then use these relations to show the statement of the lemma.

First, let us focus on $\NEQ$. We aim to use Lemma~\ref{lem:newRelation}.
Let $\calJ^{\NEQ}$ be a five-vertex path $(u - w_1 - w_2 - w_3 - v)$.
The interface of $\calJ^{\NEQ}$ is $\boldx= (u,v)$.
The lists are $L(u) = L(v) = L(w_2) = \{a,c\}$, and $L(w_1) = L(w_3) = \{b,d\}$.

Let us analyze the values of $\hom{(\calJ^{\NEQ},\boldx)}{(P_4,\boldf)}$ for distinct $\boldf \in V(P_4)^2$:
\begin{align*}
&\hom{(\calJ^{\NEQ},\boldx)}{(P_4,(a,a))} =   2,\\
&\hom{(\calJ^{\NEQ},\boldx)}{(P_4,(a,c))}  =   3,\\
&\hom{(\calJ^{\NEQ},\boldx)}{(P_4,(c,a))}  =   3,\\
&\hom{(\calJ^{\NEQ},\boldx)}{(P_4,(c,c))}  =   5.\\
&\hom{(\calJ^{\NEQ},\boldx)}{(P_4,\boldf)}  =   0\text{, otherwise}.
\end{align*}
So, for each $\boldf\in \NEQ$ and $\boldg\in V^2\setminus \NEQ$,
we have $\hom{(\calJ, (s,t))}{(\calH',\boldf)}=3$ and $\hom{(\calJ, (s,t))}{(\calH',\boldg)}\in \{0,2,5\}$.
Thus, $\calJ^{\NEQ}$ satisfies the assumptions of Lemma~\ref{lem:newRelation} and we obtain that $\NEQ$ is realizable by $P_4$.

Now let us consider $\OR_q$.
Define $\calJ^{\OR_q}$ as follows. We introduce $q$ vertices $v_1,v_2,\ldots,v_q$ and one additional vertex $w$, adjacent to all $v_i$'s. We set $L(v_i) = \{a,c\}$ for all $i \in [q]$, and $L(w) = \{b,d\}$.
The interface of $\calJ^{\OR_q}$ is $\boldx= (v_1,v_2,\ldots,v_q)$.

For $\boldf\in V(P_4)^q$ we have
\begin{align*}
\hom{(\calJ^{\OR_q},\boldx)}{(P_4,\boldf)}  =   0 & \text{ if } \boldf \notin \{a,c\}^q,\\
\hom{(\calJ^{\OR_q},\boldx)}{(P_4,\boldf)}  =   1 & \text{ if } \boldf \in \{a,c\}^q\setminus c^q,\\
\hom{(\calJ^{\OR_q},\boldx)}{(P_4,\boldf)}  =   2 & \text{ if } \boldf = c^q.
\end{align*}
Again, $\calJ^{\OR_q}$ satisfies the assumptions of Lemma~\ref{lem:newRelation} and we obtain that $\OR_q$ is realizable by $P_4$.

Finally, consider an arbitrary relation $R \subseteq \{a,c\}^q$.
Let $\{\boldf_1,\boldf_2,\ldots,\boldf_p\} = \{a,c\}^q \setminus R$, and for each $i \in [p]$ let $R_i := \{a,c\}^q \setminus \{\boldf_i\}$.
Clearly $R = \bigcap_{i=1}^p R_i$.

Thus, by Lemma~\ref{lem:simplecombinations}, it is sufficient  to show that each $R_i$ is realizable.
Fix some $i$ and let $J$ be the set of the indices $j \in [q]$, such that the $j$-th coordinate of $\boldf_i$ is $a$ (and the other ones are 
$c$). If $|J|=\emptyset$, then $R_i = \OR_q$ and we are done. So suppose that $|J| \geq 1$ and we can realize all relations $R' = \{a,c\}^q \setminus \{\boldf\}$, where the number of coordinates of $\boldf$ that are equal to $a$ is smaller than $|J|$.

So let us choose some $j \in J$. For each tuple $\boldf\in \{a,c\}^q$ let $\boldf'$ be the tuple in $\{a,c\}^q$ obtained from $\boldf$ by changing the $j$-th coordinate from $a$ to $c$ or from $c$ to $a$, whichever applies.
Consider $R'_i = \{a,c\}^q \setminus \{\boldf_i'\}$. Since $j\in J$, $\boldf_i'$ is obtained from $\boldf_i$ by swapping the $j$-th coordinate from $a$ to $c$. By the inductive assumption, $R'_i$ is realizable.

Note that $R_i=\{\boldf \mid \boldf'\in R_i'\}$. We will use Corollary~\ref{cor:simpleGadget} to show that $R_i$ is realizable by the structure $\calH=(V(P_4), \{E(P_4), R_i', \NEQ\})$, which implies that $R_i$ is realizable by $P_4=(V(P_4), E(P_4))$. Slightly abusing notation, we use $E$, $R_i'$, and $\NEQ$ also as corresponding relation symbols in the signature of $\calH$.
We define a gadget $\calJ$ on $q+1$ vertices $\{v_1, \dots, v_q, u\}$ with interface $\boldx=(v_1, \dots, v_q)$. We apply $R_i'$ to the tuple $(v_1,\dots v_{j-1}, u, v_{j+1}, \dots v_q)$ and we apply $\NEQ$ to $(u,v_j)$, i.e.~$R_i'^{\calJ}=\{(v_1,\dots v_{j-1}, u, v_{j+1}, \dots v_q)\}$, $\NEQ^{\calJ}=\{(u,v_j)\}$, and $E^{\calJ}=\emptyset$.
Clearly, $\calJ$ has the same signature as $\calH$. Moreover, $\hom{(\calJ, \boldx)}{(\calH, \boldf)}\neq0$ iff $\boldf\in R_i$.
\end{proof}

Note that Lemma~\ref{lem:basicP4} together with the proof of Theorem~\ref{thm:bipartitehardness} already yield the following result.

\begin{cor}\label{cor:hardnessp4}
Assuming the \#SETH, there is no $\epsilon >0$, such that $\LHom{P_4}$ on $n$-vertex instances
given with a path decomposition of width $t$ can be solved in time $(2-\epsilon)^t \cdot n^{\bigO(1)}$.
\end{cor}

\subsection{Structure of $P_4$s in $H$}
Let $H$ be a bipartite graph with bipartition $(X,Y)$.
We define
\begin{align*}
\calP_3(H) := & \{ S \subseteq V(H) \mid H[S] \simeq P_3 \},\\
\calP_4(H) := & \{ S \subseteq V(H) \mid H[S] \simeq P_4 \}.
\end{align*}

By $\pfour{H}$ we denote the following graph:
\begin{align*}
V(\pfour{H}) := & \calP_4(H),\\
E(\pfour{H}) := & \left\{ SS' \in \calP_4(H)^2 \mid (|(S \cap S') \cap X|=2 \lor |(S \cap S') \cap Y|=2)\right\}.
\end{align*}
Informally speaking, we think of two induced four-vertex paths as adjacent if they share two vertices from one bipartition class.
Let us point out that in the definition of $E(\pfour{H})$ we do not insist that $S \neq S'$.
Therefore the graph $\pfour{H}$ is reflexive, i.e., every vertex has a loop. This is a technical detail that allows us to simplify some arguments.

\begin{lem}\label{lem:p4connected}
Let $H$ be a connected irredundant bipartite graph.
Then $\pfour{H}$ is connected.
\end{lem}
\begin{proof}
Note that if $|V(H)| \leq 3$, then $\pfour{H} = \emptyset$ and we are done.
Thus, suppose that $H$ has at least four vertices.
Let $C_1,C_2,\ldots,C_p$ be the connected components of $\pfour{H}$.
Let $c \from \calP_4(H) \to [p]$ be a function such that $c(S)=i$ if and only if $S$ belongs to $C_i$.

First, let us observe the following.
\begin{clm} \label{clm:inP4}
Every edge and every induced $P_3$ in $H$ is contained in some induced $P_4$.
\end{clm}
\begin{claimproof}
First, let us argue that every induced $P_3$ is contained in some induced $P_4$.
Let $(x-y-z)$ be an induced $P_3$. Since the vertices $x$ and $z$ do not have the same neighborhood, one of them, say $x$,
has a neighbor $w$ that is not adjacent to $z$. Then $(w-x-y-z)$ is an induced $P_4$.
Now consider an edge $xy$. Since $H$ is connected and has at least three vertices, one of the endvertices of $xy$, say $y$, has a neighbor $z$.
Since $H$ is bipartite, $(x-y-z)$ is an induced $P_3$. As we have already shown, it is contained in some induced $P_4$.
\end{claimproof}

Note that if $A \in \calP_3(H)$, then the set $\calS_A := \{S \in \calP_4(H) \mid A \subseteq S\}$ is non-empty (by Claim~\ref{clm:inP4}) and induces a clique in $\pfour{H}$.
In particular, for any $S,S' \in \calS_A$ it holds that $c(S)=c(S')$.
We introduce a labeling $\ell \from \calP_3(H) \to [p]$, where $\ell(A) = i$ if and only if $c(\calS_A) = \{i\}$.

\begin{clm}\label{clm:edgelabels}
For any edge $xy \in E(H)$ and any $A,B \in \calP_3(H)$ such that $\{x,y\} \subseteq A \cap B$,
it holds that $\ell(A) = \ell(B)$.
Consequently, all sets $S \in \calP_4(H)$ that contain $xy$ belong to the same connected component of $\pfour{H}$.
\end{clm}
\begin{claimproof}
Let $A = \{x,y,a\}$ and $B= \{x,y,b\}$, where $a \neq b$.
If $\{x,y,a,b\} \in \calP_4(H)$, then $\{x,y,a,b\} \in \calS_{A} \cap \calS_B$ and thus $\ell(A) = \ell(B) = c(\{x,y,a,b\})$.

So suppose that $\{x,y,a,b\} \notin \calP_4(H)$. This is possible in two cases: (1) if $a$ and $b$ are adjacent to the same vertex from $\{x,y\}$ (and thus $H[\{x,y,a,b\}]$ is a star with $3$ leaves), or (2) if $ab \in E(H)$ (and thus $H[\{x,y,a,b\}]$ is a $4$-cycle).

By Claim~\ref{clm:inP4}, each of the sets $A,B$ belongs to some induced $P_4$, 
i.e., there are vertices $c,d$, such that $\{x,y,a,c\}, \{x,y,b,d\} \in \calP_4(H)$.
Note that it is possible that $c=d$.

We will show that there is a walk from $\{x,y,a,c\}$ to $\{x,y,b,d\}$ in $\pfour{H}$,
which will prove that $c(\{x,y,a,c\})=c(\{x,y,b,d\})$ and consequently $\ell(A) = \ell(B)$.
We consider some cases.
\begin{enumerate}[{Case }1.]
\item Suppose that $H[\{x,y,a,b\}]$ is a star with $3$ leaves. By symmetry we assume that $a,b,x \in N(y)$.
We consider possible positions of $c$ and $d$. One of the following cases occurs, see also Figure~\ref{fig:p4case1}.
\begin{enumerate}[{Subcase 1.}1.]
\item $N(x) \cap \{c,d\} = 0$.
If none of the edges $cb, ad$ exists in $H$, then we have a walk $\bigl(\{x,y,a,c\}-\{c,a,y,b\}-\{a,y,b,d\}-\{x,y,b,d\}\bigr)$ in $\pfour{H}$.
If at least one of these edges, say $cb$, exists in $H$, then we have a walk $\bigl(\{x,y,a,c\}-\{x,y,b,c\}-\{x,y,b,d\}\bigr)$ in $\pfour{H}$.

\item $N(x) \cap \{c,d\} = 1$. By symmetry, assume that $dx \in E(H)$.
If $ad \notin E(H)$, then we have a walk $\bigl(\{x,y,a,c\}-\{d,x,y,a\}-\{d,x,y,b\}\bigr)$ in $\pfour{H}$.
If $cb \in E(H)$, then we have a walk $\bigl(\{x,y,a,c\}-\{x,y,b,c\}-\{d,x,y,b\}\bigr)$ in $\pfour{H}$.
If $cb \notin E(H)$ and $ad \in E(H)$, then we have a walk $\bigl(\{x,y,a,c\}-\{c,a,y,b\}-\{d,a,y,b\}-\{d,x,y,b\}\bigr)$ in $\pfour{H}$.

\item $N(x) \cap \{c,d\} = 2$. 
If one of the edges $cb, ad$, say $ad$, does not exist in $H$, then we have a walk $\bigl(\{c,x,y,a\}-\{d,x,y,a\}-\{d,x,y,b\}\bigr)$ in $\pfour{H}$.
If both $cb,ad \in E(H)$, then we have a walk $\bigl(\{c,x,y,a\} - \{c,x,d,a\}-\{d,x,c,b\}-\{d,x,y,b\}\bigr)$ in $\pfour{H}$.
\end{enumerate}

\begin{figure}
\begin{center}
\includegraphics[scale=0.78,page=1]{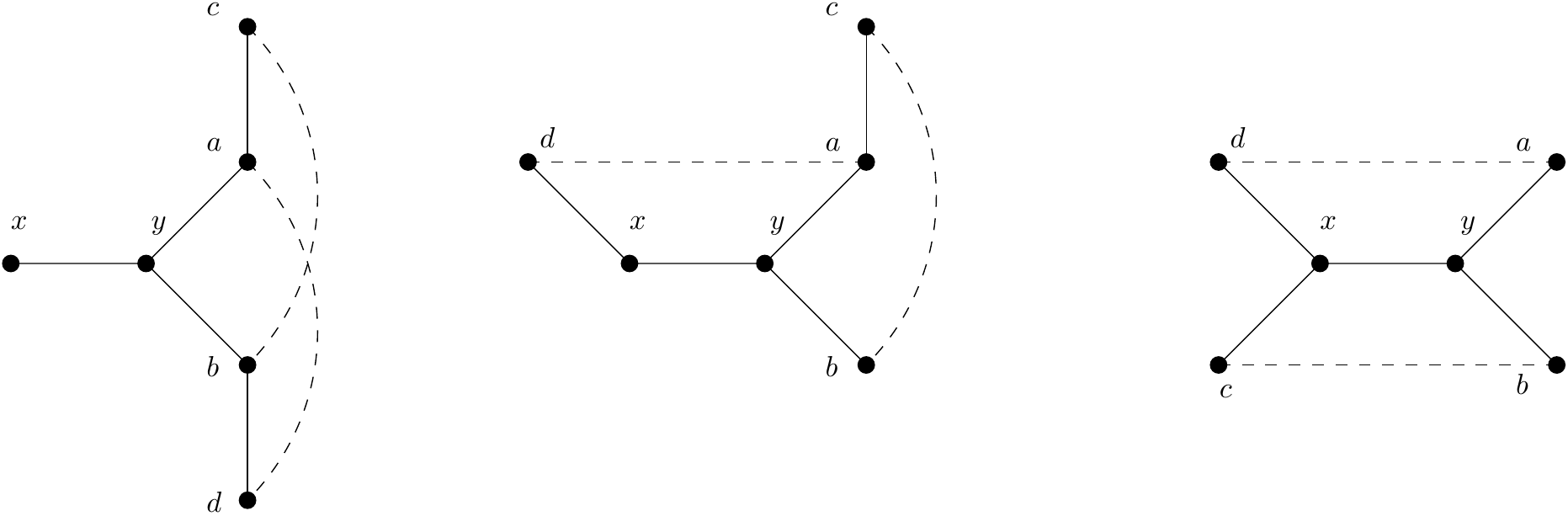}
\end{center}
\caption{Possible configurations in the proof of Case 1 in Claim~\ref{clm:edgelabels}. Dashed edges may, but do not have to exist.
Note that the first and last subcase vertices $c$ and $d$ may coincide.} \label{fig:p4case1}
\end{figure}

\item Suppose that $H[\{x,y,a,b\}]$ is a $4$-cycle. By symmetry we assume that consecutive vertices of the cycle are $a,x,y,b$.
Note that $c$ is adjacent to exactly one of $a, y$, and $d$ is adjacent to exactly one of $x, b$, see also Figure~\ref{fig:p4case2}.

\begin{enumerate}[{Subcase 2.}1.]
\item $dx, cy \in E(H)$.
If $cd \notin E(H)$, then we have a walk $\bigl(\{a,x,y,c\}-\{d,x,y,c\}-\{d,x,y,b\}\bigr)$ in $\pfour{H}$.
If $cd \in E(H)$, then we have a walk $\bigl(\{a,x,y,c\}-\{a,x,d,c\}-\{b,a,x,d\}-\{b,y,x,d\}\bigr)$ in $\pfour{H}$.

\item Exactly one of the edges $xd, yc$, say $xd$ exists in $H$.

If $cd \notin E(H)$, then we have a walk $\bigl(\{c,a,x,y\}-\{c,a,x,d\}-\{b,a,x,d\}\bigr)$ in $\pfour{H}$.
If $cd \in E(H)$, then we have a walk $\bigl(\{c,a,x,y\}-\{c,d,x,y\}-\{d,x,y,b\}\bigr)$ in $\pfour{H}$.

\item $ac, bd \in E(H)$. 
If $cd \notin E(H)$, then we have a walk $\bigl(\{c,a,x,y\}-\{c,a,b,y\}-\{c,a,b,d\}-\{x,a,b,d\}-\{x,y,b,d\}\bigr)$ in $\pfour{H}$.
If $cd \in E(H)$, then we have a walk $\bigl(\{c,a,x,y\}-\{x,a,c,d\}-\{x,a,b,d\}-\{x,y,b,d\}\bigr)$ in $\pfour{H}$.
\end{enumerate}
\end{enumerate}

\begin{figure}
\begin{center}
\includegraphics[scale=0.8,page=2]{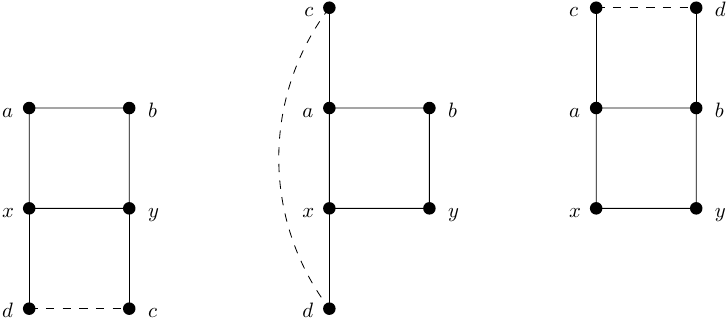}
\end{center}
\caption{Possible configurations in the proof of Case 2 in Claim~\ref{clm:edgelabels}. Dashed edges may, but do not have to exist.} \label{fig:p4case2}
\end{figure}
As these are all possible cases, the proof of the claim is complete.
\end{claimproof}

Claim~\ref{clm:edgelabels} allows us to define a labeling $\ell'$ of edges of $H$, analogous to $\ell$:
for $xy \in E(H)$, we set $\ell'(xy) = c(S)$, where $S$ is any element of $\calP_4(H)$ with $\{x,y\} \in S$.
Note that by Claim~\ref{clm:inP4}, such a set $S$ always exists and by Claim~\ref{clm:edgelabels}, the value of $\ell'(xy)$ does not depend on the choice of $S$. Note that for every induced path $(x-y-z)$ we have $\ell(\{x,y,z\}) = \ell'(xy)=\ell'(yz)$.

Now we claim that for any two edges $e,f \in E(H)$ we have $\ell'(e) = \ell'(f)$.
For contradiction, suppose this is not the case. Since $H$ is connected, this means that there are two edges $xy$ and $xz$, where $y \neq z$, such that $\ell'(xy) \neq \ell'(xz)$. However, this cannot happen as $\ell'(xy)=\ell(\{x,y,z\})=\ell'(yz)$, since $\{x,y,z\} \in \calP_3(H)$.
\end{proof}

\subsection{Constructing $(x,y)$-Forcers}\label{sec:Forcers}
\begin{defn}[$(x,y,S)$-distinguisher, $(x,y,S)$-forcer]\label{def:xyDistinguisher}
Let $H=(V,E)$ be a bipartite graph. Let $S$ be a one-sided subset of $V$ with distinct vertices $x,y \in S$.
Let $\{\alpha,\beta\}$ be a one-sided subset of $V$, and let $R$ be a relation in $S\times \{\alpha,\beta\}$.
\begin{enumerate}
\item $R$ is an \emph{$(x,y,S)$-distinguisher with respect to $(\alpha,\beta)$} if it has the following properties:
\begin{myitemize}
\item $(x,\alpha)\in R$, $(x,\beta) \notin R$.
\item $(y,\beta)\in R$.	
\item $\forall v \in S : R\cap\{(v,\alpha), (v,\beta)\}\neq \emptyset$.		
\end{myitemize}
\item $R$ is an \emph{$(x,y,S)$-forcer with respect to $(\alpha,\beta)$} if it is an $(x,y,S)$-distinguisher with respect to $(\alpha,\beta)$ with the additional property:
\begin{myitemize}
\item $(y,\alpha)\notin R$.
\end{myitemize}
\end{enumerate} 
We say that a relation $R$ is an \emph{$(x,y,S)$-distinguisher} (resp.~\emph{$(x,y,S)$-forcer}) if there are is a one-sided set $\{\alpha,\beta\}$,
such that $R$ is an $(x,y,S)$-distinguisher (resp.~$(x,y,S)$-forcer) with respect to $(\alpha,\beta)$.
\end{defn}

The following lemma is a crucial building block that we will use repeatedly in the results leading up to the proof of Lemma~\ref{lem:allrelations}. Here, given a realizable $(x,y,S)$-distinguisher with respect to some pair of vertices on a four-vertex path, we use Lemma~\ref{lem:primes} to turn this distinguisher into a forcer.

\begin{lem}\label{lem:DistinguisherToForcer}
Let $H=(V,E)$ be an irredundant connected bipartite graph.
Let $S$ be a one-sided subset of $V$ with distinct vertices $x,y \in S$.
Let $R$ be a relation that is realizable by $H$ such that, for some induced $4$-vertex path $(a-b-c-d)$ in $H$, $R$ is an $(x,y,S)$-distinguisher with respect to $(a,c)$ or with respect to $(c,a)$.
Then the following relations are realizable by $H$.
\begin{myenumerate}
\item an $(x,y,S)$-forcer with respect to $(a,c)$, \label{it:ac}
\item an $(x,y,S)$-forcer with respect to $(c,a)$, \label{it:ca}
\item an $(x,y,S)$-forcer with respect to $(b,d)$. \label{it:bd}
\item an $(x,y,S)$-forcer with respect to $(d,b)$, \label{it:db}
\end{myenumerate}
\end{lem}
\begin{proof}
First note that, since $R$ is realizable by $H$, there is a pathwidth-preserving reduction from $\LHom{(V, E\cup\{R\})}$ to $\LHom{H}$. So it suffices to show that the relations in items~\ref{it:ac}-\ref{it:db} are realizable by $\calH'=(V, E\cup\{R\})$.
	
Recall that by Lemma~\ref{lem:basicP4}, the relations $\NEQ(\{a,c\})=\{(a,c),(c,a)\}$  and $\NEQ(\{b,d\})=\{(b,d),(d,b)\}$ are realizable by $P_4$ and thus by $H$.
Note that if $R$ is an $(x,y,S)$-distinguisher with respect to $(c,a)$, then $R \con \NEQ(\{a,c\})$ is an $(x,y,S)$-distinguisher  with respect to $(a,c)$. Furthermore, by Lemma~\ref{lem:simplecombinations}, this relation is realizable by $\calH'$.
So from now on we can assume that $R$ is an $(x,y,S)$-distinguisher with respect to $(a,c)$, realizable by $\calH'$.

Let us first focus on proving item~\ref{it:ac}.
\begin{clm} \label{clm:forcerac}
Let $R$ is an $(x,y,S)$-distinguisher with respect to $(a,c)$, realizable by $\calH'$.
Then there is an $(x,y,S)$-forcer $R'$ with respect to $(a,c)$, realizable by $\calH'$.
\end{clm}
\begin{claimproof}
Note that if $(y,a)\notin R$ then we can choose $R'=R$ and are done. So we can assume that $(y,a)\in R$.
	
Slightly overloading notation, we use $E$ and $R$ also as the relation symbols in $\sigma(\calH')$ corresponding to the relations $E$ and $R$, respectively.
We now define a gadget $\calJ=(U, \boldS)$ with two interface vertices $s,t$ that has the same signature as $\calH'$.
The vertices of $\calJ$ are $U=\{s,t,t', u_1, u_2, u_3\}$, the edge relation is $E^\calJ=E((t-u_1- t'- u_2- u_3))$, and $R^\calJ=\{(s,t), (s,t')\}$. As lists we set $L(u_1)=L(u_2)=\{b,d\}$ and $L(u_2)=\{a,c\}$. Intuitively, $\calJ$ is a path $(t-u_1- t'-u_2- u_3)$, where in addition $s,t$ and $s,t'$ are subject to the relation $R$. The gadget $\calJ$ is illustrated in the following figure:
	
\begin{center} 
\begin{tikzpicture}
	[scale=0.8, node distance = 1.4cm]
	\tikzstyle{vertex}=[fill=black, draw=black, circle, inner sep=2pt]

	% Main sequence of vertices 
	\node[vertex, label={[label distance=.25cm]180:$s$}] (s) at (0  ,0)  {};
	\node[vertex, label={[label distance=.25cm]0:$t$}] (t) at (4  ,1.5) {};
	\node[vertex,label={[label distance=.25cm]270:$t'$}] (tt) at (4  ,-1.5)  {};
	\node[vertex, label={[label distance=.25cm]0:$u_1$}, label={[label distance=1cm]0:$\rightarrow \{b,d\}$}] (u1) at (4  ,0)  {};
	\node[vertex, label={[label distance=.25cm]270:$u_2$},  label={[label distance=.74cm]270:$\rightarrow \{b,d\}$}] (u2) at (6  ,-1.5)  {};
	\node[vertex, label={[label distance=.25cm]270:$u_3$}, label={[label distance=.75cm]270:$\rightarrow \{a,c\}$}] (u3) at (8  ,-1.5)  {};
			
	\draw (s) -- ++(t) node [midway,fill=white, rectangle, draw=black] {$R$};
	\draw (s) -- ++(tt) node [midway,fill=white, rectangle, draw=black] {$R$};
	\draw (t) -- (u1) --(tt) --(u2) --(u3);
	%			\node () at (2,-2) {$\calJ$};
\end{tikzpicture}
\end{center}
	
Consider the following subsets of $S$.
\begin{myitemize}
	\item $S_a=\{v\in S \mid (v,a) \in R, (v,c)\notin R\}$,
	\item $S_c=\{v\in S \mid (v,a) \notin R, (v,c)\in R\}$,
	\item $S_{ac}=\{v\in S \mid (v,a), (v,c)\in R\}$.
\end{myitemize}
Note that $R= (S_a\times \{a\}) \cup (S_c\times \{c\}) \cup (S_{ac}\times \{a,c\})$, $x\in S_a$, and $y\in S_{ac}$. We choose $R'=R\setminus\{(v,a) \mid v\in S_{ac}\}$, i.e., $R'= (S_a\times \{a\}) \cup (S_c\times \{c\}) \cup (S_{ac}\times \{c\})$. Note that $R'$ is an $(x,y,S)$-forcer with respect to $(a,c)$. Let us now show that $R'$ is realizable by $\calH'$.
For each tuple $(v,w)\in V^2$, we determine $\hom{(\calJ, (s,t))}{(\calH',(v,w))}$:
\begin{itemize}
	\item If $v\in S_a$, $\hom{(\calJ, (s,t))}{(\calH',(v,a))}=2$.
	\item For $v\in S_c$, $\hom{(\calJ, (s,t))}{(\calH',(v,c))}=6$.
	\item For $v\in S_{ac}$, $\hom{(\calJ, (s,t))}{(\calH',(v,a))}=5$.
	\item For $v\in S_{ac}$, $\hom{(\calJ, (s,t))}{(\calH',(v,c))}=8$.
	\item Otherwise, $\hom{(\calJ, (s,t))}{(\calH',(v,w))}=0$.
\end{itemize}
So, for each $\boldf\in R'$ and $\boldg\in V^2\setminus R'$,
\begin{enumerate}
	\item $\hom{(\calJ, (s,t))}{(\calH',\boldf)}\in \{2,6,8\}$.
	\item  $\hom{(\calJ, (s,t))}{(\calH',\boldg)}\in \{0,5\}$.
\end{enumerate}
Therefore, we can apply Lemma~\ref{lem:primes} to complete the proof. 
\end{claimproof}

Let $R'$ be the $(x,y,S)$-forcer with respect to $(a,c)$ given by Claim~\ref{clm:forcerac}.
To obtain item~\ref{it:ca}, notice that $R' \con \NEQ(\{a,c\})$ is an $(x,y,S)$-forcer with respect to $(c,a)$, and by Lemma~\ref{lem:simplecombinations} it is realizable by $\calH'$.

Now let us focus on \eqref{it:bd}. Note that $R' \con R_{\{a,c\} \to \{b,d\}}$ is an $(x,y,S)$-distinguisher with respect to $(b,d)$, which is realizable by $\calH'$ by Lemma~\ref{lem:simplecombinations}. Thus, applying Claim~\ref{clm:forcerac} with $a$ switched to $b$ and $c$ switched to $d$, we obtain that an $(x,y,S)$-forcer $R''$ with respect to $(b,d)$ is realizable by $\calH'$.

Finally, to show \eqref{it:db}, note that $R'' \con \NEQ(\{b,d\}$ is an $(x,y,S)$-forcer with respect to $(d,b)$, realizable by $\calH'$ (by Lemma~\ref{lem:simplecombinations}). This completes the proof.
\end{proof}

We show a strengthening of Lemma~\ref{lem:DistinguisherToForcer}: Given a realizable $(x,y,S)$-distinguisher with respect to some pair of vertices that are potentially far apart in $H$, we can obtain a realizable $(x,y,S)$-forcer with respect to some pair of vertices on an induced four-vertex path. Moreover, since $P_4$s form a connected structure in $H$ (recall Lemma~\ref{lem:p4connected}), we can even choose a $P_4$ in $H$ and obtain a forcer with respect to any one-sided pair from this very $P_4$.

\begin{lem}\label{lem:forceronP4}
Let $H=(V,E)$ be an irredundant connected bipartite graph.
Fix an induced 4-vertex path $(a-b-c-d)$ in $H$.
Let $S$ be a one-sided subset of $V$ with distinct $x,y\in S$. 
Suppose there is an $(x,y,S)$-distinguisher $R'$ that is realizable by $H$.
Then there is an $(x,y,S)$-forcer $R$ with respect to $(a,c)$ such that $R$ is realizable by $H$.
\end{lem}
\begin{proof}
The proof is divided into two parts.
First, we show that there is a realizable forcer with respect to a one-sided pair of vertices on \emph{some} induced $P_4$ in $H$.

\begin{clm} \label{clm:forceronP4}
There exist an induced $4$-vertex path $(a' - b' - c' -d' )$ in $H$ and an $(x,y,S)$-forcer with respect to $(a',c')$ that is realizable by $H$.
\end{clm}
\begin{claimproof}
Let $\{\alpha,\beta\}$ be a one-sided set such that $R'$ is an $(x,y,S)$-distinguisher with respect to $(\alpha,\beta)$.
If $\alpha$ and $\beta$ have a common neighbor then, since $H$ is irredundant and has at least four vertices, there is an induced $4$-vertex path in $H$ that is of the form $(\alpha -b' - \beta -d')$ or of the form $(d' - \alpha -b' - \beta)$. In either case, the statement follows from Lemma~\ref{lem:DistinguisherToForcer}.

Otherwise, there is a shortest path $P=(p_1 - \ldots - p_k)$ from $\alpha$ to $\beta$ ($p_1= \alpha$, $p_k= \beta$) with $k$ vertices. Since $\alpha$ and $\beta$ are in the same bipartition class we have $k\ge 5$. For each $i\in \{1,\ldots, k-4\}$, the relation
\[
R_i= R_{\{p_i, p_k\}\to \{p_{i+1}, p_{k-1}\}}\con R_{\{p_{i+1}, p_{k-1}\}\to \{p_{i+2}, p_{k}\}}
\]
is a $(p_i, p_k, \{p_i, p_k\})$-forcer with respect to $(p_{i+2}, p_k)$, and $R_i$ is realizable by $H$ according to Lemma~\ref{lem:simplecombinations}.
So, by Lemma~\ref{lem:simplecombinations}, the relation $R'\con R_1\con \ldots\con R_{k-4}$ is a realizable $(x,y,S)$-distinguisher with respect to $(p_{k-2},p_k)$, where $(p_{k-3} - p_{k-2} - p_{k-1} - p_k)$ is an induced path in $H$.
By Lemma~\ref{lem:DistinguisherToForcer}, there is also a realizable $(x,y,S)$-forcer with respect to $(p_{k-2},p_k)$.
\end{claimproof}

Let $(a'-b'-c'-d')$ be as in Claim~\ref{clm:forceronP4}.
Note that $H$ satisfies the assumptions of Lemma~\ref{lem:p4connected}, so there is a sequence $P^{(1)}, P^{(2)}, \ldots, P^{(s)}$ of subsets of $H$, such that (i) each $P^{(i)}$ induces a $P_4$ in $H$, (ii) $P^{(1)} = \{a',b',c',d'\}$ and $P^{(s)}=\{a,b,c,d\}$, and (iii) for each $i \leq s-1$, sets $P^{(i)}$ and $P^{(i+1)}$ share two vertices from the same bipartition class.

We prove the statement by induction on $s$.
If $s=1$, then either $(a-b-c-d) = (a'-b'-c'-d')$ or $(a-b-c-d) = (d'-c'-b'-a')$.
By Lemma~\ref{lem:DistinguisherToForcer}, in either case there is a realizable $(x,y,S)$-forcer with respect to $(a,c)$ and we are done.

So assume that $s \geq 2$ and there is an $(x,y,S)$-forcer $R''$ with respect to $(a'',c'')$, such that $R''$ is realizable by $H$,
where $H[P^{(s-1)}]=(a''-b''-c''-d'')$.
Let $(s,t) \in \{ (a'',c''), (b'',d'') \}$ be such that $s,t \in P^{s-1} \cap P^{s}$.
Applying Lemma~\ref{lem:DistinguisherToForcer} for the path $H[P^{s-1}]$, we obtain that an $(x,y,S)$-forcer with respect to $(s,t)$ is realizable by $H$. Since $s,t\in P^{s}$, where $s$ and $t$ are in the same bipartition class, we have $(s,t)\in \{(a,c), (c,a), (b,d), (d,b)\}$.
Now using Lemma~\ref{lem:DistinguisherToForcer} for the path $H[P^{s}]=(a-b-c-d)$, we conclude that an $(x,y,S)$-forcer with respect to $(a,c)$ is realizable by $H$.
\end{proof}

Recall from Lemma~\ref{lem:basicP4} that the structure of a $4$-vertex path is rich enough to encode basic binary relations.
In Lemma~\ref{lem:forceronP4} we showed how to obtain forcers with respect to some specified pair $(a,c)$ on a $4$-vertex path. In the next lemma we show that such a collection of forcers lets us realize more expressive relations.

\begin{lem} \label{lem:allrelationsSac}
Let $H=(V,E)$ be an irredundant connected bipartite graph.
Let $(a-b-c-d)$ be an induced $4$-vertex path in $H$, and let $S$ be a one-sided subset of $V$. 
If for every $x,y \in S$ there is an $(x,y,S)$-distinguisher realizable by $H$,
then, for every fixed $p \geq 1$ and $q \geq 0$, every relation $R \subseteq S^p \times \{a,c\}^q$ is realizable by $H$.
\end{lem}
\begin{proof}
Let $|S|=s$ and enumerate $S$ as $\{x_1,x_2,\ldots,x_s\}$.
Let $(a-b-c-d)$ be a fixed $4$-vertex path in $H$. 
By the assumption of the lemma, for all distinct $x_i,x_j \in S$ there is some $(x_i,x_j,S)$-distinguisher realizable by $H$.
Thus, by Lemma~\ref{lem:forceronP4}, there is also an $(x_i,x_j,S)$-forcer $R_{i,j}$ with respect to $(a,c)$, realizable by $H$.

First, let us define an auxiliary relation $I \subseteq S \times \{a,c\}^{s(s-1)}$ and show that it is realizable by $H$.
Recall from Section~\ref{sec:preliminaries} that $I(x)=\{(v_1,\ldots, v_{s(s-1)}) \in \{a,c\}^{s(s-1)} \mid (x, v_1,\ldots, v_{s(s-1)}) \in I\}$.
We say that a relation $I$ is an \emph{indicator} if
\begin{itemize}
\item $I(x) \neq \emptyset$\, for all $x \in S$, and
\item $I(x) \cap I(x') = \emptyset$ for all distinct $x,x' \in S$.
\end{itemize}
Intuitively, each element of $I(x)$ uniquely represents the value $x$, and every $x \in S$ admits such a representation.

\begin{clm}\label{clm:indicator}
There is an indicator $I \subseteq S \times \{a,c\}^{s(s-1)}$ realizable by $H$.
\end{clm}
\begin{claimproof}
For brevity, denote $\mathsf{Ind} := \{(i,j) ~|~ i,j \in [s], i \neq j \}$.
As for each $(i,j) \in \mathsf{Ind}$, the forcer $R_{i,j}$ is realizable by $H$,
it is sufficient to show an indicator realizable by the structure $\calH' = (V, \bigcup_{(i,j) \in \mathsf{Ind}} \{R_{i,j}\} )$.

We define a structure $\calI = (U,\boldS)$ with the same signature as $\calH'$.
The vertex set of $\calI$ is $U = \{u\} \cup \bigcup_{(i,j) \in \mathsf{Ind}}\{u_{i,j}\}$.
We also define an interface $\boldu$ of $\calI$ by setting
\[\boldu = (u, u_{1,2}, \ldots, u_{1,s}, u_{2,1}, u_{2,3}, \ldots, u_{2,s}, \ldots, u_{s,1}, \ldots, u_{s,s-1}).\]
For all $(i,j) \in \mathsf{Ind}$, we set $R_{i,j}^{\calI} = \{ (u,u_{i,j}) \}$, i.e., the $(x_i, x_j, S)$-forcer $R_{i,j}$ is applied to the tuple $(u, u_{i,j})$.
Consider the set
\begin{align*}
I:=  \{ \boldf \in V^{s(s-1)+1} \mid \hom{(\calI, \boldu)}{(\calH', \boldf)} \neq 0\}
\end{align*}
and let $(f,f_{1,2},\ldots,f_{1,s}, f_{2,1}, f_{2,3}, \ldots, f_{2,s}, \ldots, f_{s,1}, \ldots, f_{s,s-1}) \in I$.

First, clearly $f \in S$ and for all $(i,j) \in \mathsf{Ind}$ we have $f_{i,j} \in \{a,c\}$.
Next, if $f = x_i$ for $i \in [s]$, then the properties of forcers imply that
\begin{itemize}
\item $f_{i,j} = a$ for all $j$, such that $(i,j) \in \mathsf{Ind}$, and
\item for all $i' \in [s] \setminus \{i\}$ there is $j$, such that $(i',j) \in \mathsf{Ind}$ and $f_{i',j} = c$.
\end{itemize} 
Therefore $I(x_i) \cap I(x_{i'}) = \emptyset$ for all distinct $i,i' \in [s]$.
Finally, by the definition of a forcer, we note that for all $i \in [s]$ and all $(i',j') \in \mathsf{Ind}$ we have that at least one of $(x_i,a)$ or $(x_i,c)$ is in $R_{i',j'}$, and thus $I(x_i) \neq \emptyset$.
Summing up, $I$ is an indicator and by Corollary~\ref{cor:simpleGadget} it is realizable by $H$.
Note that the exact definition of $I$ depends on the exact definitions of $R_{i,j}$ for $(i,j) \in \mathsf{Ind}$.
\end{claimproof}

Let $I$ be the indicator given by Claim~\ref{clm:indicator}; it is realizable by $H$.
For each $i \in [s]$, we fix some element $\mathsf{id}(x_i)$ from $I(x_i)$. 
We think of $\mathsf{id}(x_i)$ as a unique identifier of $x_i$. The properties of $I$ ensure these identifiers are pairwise distinct. Note that $\mathsf{id}(x_i) \in \{a,c\}^{s(s-1)}$.

Let $p \geq 1$ and $q \geq 0$ be integers and consider a relation $R \subseteq S^p \times \{a,c\}^q$.
Let $R_I \subseteq \{a,c\}^{p(s-1)s+q}$ be the relation defined as follows:
\[
R_I = \{  (\mathsf{id}(y_1),\mathsf{id}(y_2),\ldots,\mathsf{id}(y_p),b_1,b_2,\ldots,b_q) ~|~ (y_1,y_2,\ldots,y_p,b_1,b_2,\ldots,b_q) \in R \}.
\]
Intuitively, we may see $R_I$ as $R$ ``translated'' to the the ground set $\{a,c\}$, where the translation of each element of $S$ to a sequence over $\{a,c\}$ is provided by the choice of $\mathsf{id}(\cdot)$.
By Lemma~\ref{lem:basicP4}, the relation $R_I$ is realizable by $H$.
Thus, in order to prove that $R$ is realizable by $H$, it is sufficient to show that $R$ is realizable by a structure $\calH' = (V, \{ I, R_I \})$. Slightly abusing notation we use $\{ I, R_I \}$ also to denote the signature of $\calH'$.
Consider a structure $\calR = (U, \boldS)$ with signature $\{I, R_I \}$, defined as follows:
\begin{align*}
U = & \bigcup_{\ell \in [p]} \{y_\ell, u^{\ell}_1,\ldots, u^{\ell}_{s(s-1)} \} \cup \bigcup_{\ell \in [q]} \{b_\ell\},\\
I^{\calR} = & \bigcup_{\ell \in [q]} \{ (y_\ell, u^{\ell}_1,\ldots, u^{\ell}_{s(s-1)}) \},\\
R_I^{\calR} = & \{  (u^{1}_1,\ldots, u^{1}_{s(s-1)}, u^{2}_1,\ldots, u^{2}_{s(s-1)}, \ldots, u^{p}_1,\ldots, u^{p}_{s(s-1)}, b_1, \ldots, b_q) \}.
\end{align*}
Define $\boldy = (y_1,y_2,\ldots,y_p, b_1, \ldots, b_q)$.
It is straightforward to verify that $\hom{(\calR, \boldy)}{(\calH', \boldf)} \neq 0$ if and only if
$(\mathsf{id}(y_1),\mathsf{id}(y_2),\ldots,\mathsf{id}(y_p), b_1, \ldots, b_q ) \in R_I$, which in turn happens if and only if $(y_1,y_2,\ldots,y_p,b_1,\ldots,b_q) \in R$.
Thus by Corollary~\ref{cor:simpleGadget} the relation $R$ is realizable by $H$.
\end{proof}

Let us now introduce one more special type of relation.
Consider a bipartite graph $H=(V,E)$ and a subset $S\subseteq V$.
Let $(X,Y)$ be a partition of $S$. Moreover, let $a,c$ be two distinct vertices in $V$.
A relation $R \subseteq S \times \{a,c\}$ is an \emph{$(X,Y)$-partitioner with respect to $(a,c)$} if:
\begin{itemize}
\item for every $x \in X$ it holds that $R(x) = \{a\}$,
\item for every $y \in Y$ it holds that $R(y) = \{c\}$.
\end{itemize}

As each partitioner is a relation in $S \times \{a,c\}$, Lemma~\ref{lem:allrelationsSac} immediately yields the following.

\begin{cor} \label{cor:partitioners}
Let $H=(V,E)$ be an irredundant connected bipartite graph.
Let $(a-b-c-d)$ be an induced $4$-vertex path in $H$, and let $S$ be a one-sided subset of $V$. 
If for every $x,y \in S$ there is an $(x,y,S)$-distinguisher realizable by $H$,
then for every partition $(X,Y)$ of $S$, the $(X,Y)$-partitioner with respect to $(a,c)$ is realizable by $H$.
\end{cor}

\subsection{Proof of Lemma~\ref{lem:allrelations}}\label{sec:Distinguisher} \label{sec:allrelations}
In this section we finally prove Lemma~\ref{lem:allrelations}.
Let us first discuss the plan.
%Lemma~\ref{lem:allrelations} will easily follow from Lemma~\ref{lem:allrelationsSac} (for $q=0$), provided that
%there is some induced 4-vertex path $(a-b-c-d)$ in $H$ such that for all distinct $x,y \in S$, there is an $(x,y,S)$-forcer with respect to $(a,c)$ realizable by $H$.
Lemma~\ref{lem:allrelations} will easily follow from Lemma~\ref{lem:allrelationsSac} (for $q=0$), provided that
for all distinct $x,y \in S$, there is an $(x,y,S)$-distinguisher realizable by $H$.
We prove this statement inductively, essentially deriving new forcers from forcers that are ``smaller'' with respect to some measure on $(x,y,S)$. On the way, we use Lemma~\ref{lem:forceronP4} to turn distinguishers into forcers, and we use Corollary~\ref{cor:partitioners} to turn forcers into partitioners.

\begin{lem}
\label{lem:Distinguisher}
Let $H=(V,E)$ be an irredundant connected bipartite graph with $\irr(H)\ge 2$.
Let $S$ be a one-sided subset of $V$, and let $x, y$ be distinct vertices in $S$.
Fix any induced $4$-vertex path $(a-b-c-d)$ in $H$. Then there is an $(x,y,S)$-forcer $R$ with respect to $(a,c)$ such that $R$ is realizable by $H$.
\end{lem}
\begin{proof}
Let $S$ and $S'$ be one-sided irredundant sets (not necessarily in the same part of the bipartition) such that $S$ contains two distinct vertices $x,y\in S$ and $S'$ contains two distinct vertices $x',y'\in S'$. We define an order by setting $(x',y',S') < (x,y,S)$ if one of the following holds:
\begin{myitemize}
\item $|S'|< |S|$, or
\item $|S'|=|S|$ and $\dist(\{x',y'\}, S'\setminus\{x',y'\})< \dist(\{x,y\}, S\setminus\{x,y\})$.
\end{myitemize} 
We prove the statement of the lemma by induction with respect to this order. Note that $(x,y,S)$ is a minimal element if $S=\{x,y\}$.
	
Suppose that the following holds:
\begin{equation}\label{eq:distinguishersuffices}
\text{There is an $(x,y,S)$-distinguisher or a $(y,x,S)$-distinguisher realizable by $H$.}
\end{equation}
In the first case we can apply Lemma~\ref{lem:forceronP4} to obtain an $(x,y,S)$-forcer with respect to $(a,c)$ that is realizable by $H$.
In the second case, note that a $(y,x,S)$-forcer with respect to $(a,c)$, obtained by the application of Lemma~\ref{lem:forceronP4}, 
is an $(x,y,S)$-forcer with respect to $(c,a)$. Thus, by Lemma~\ref{lem:DistinguisherToForcer}, we again  obtain an $(x,y,S)$-forcer with respect to $(a,c)$ that is realizable by $H$.
Therefore, in order to show the statement of the lemma, it suffices to show~\eqref{eq:distinguishersuffices}.
	
\medskip
	
We first consider the base case $S=\{x,y\}$. 
Since $H$ is irredundant we have $N(x)\neq N(y)$. Suppose there is a vertex $q\in N(y)\setminus N(x)$ (the other case is analogous). Since $H$ is connected there is some vertex $p\in N(x)$ (possibly $p\in N(x)\cap N(y)$). Since $S$ is one-sided $R_{\{x,y\}\to \{p,q\}}$ is an $(x,y,S)$-distinguisher with respect to $(p,q)$ (if $p\notin N(y)$ it is even a forcer). This is what we need according to~\eqref{eq:distinguishersuffices} and this completes the base case.

	For the inductive step, let $x,y,S$ be as given in the statement of the lemma. We can assume $S\neq \{x,y\}$. To shorten notation let $S_0=S\setminus\{x,y\}$.
	Let $P$ be a shortest path from $\{x,y\}$ to $S_0$. So $P$ is of the form $(p_1 - \ldots -p_k)$, where $p_1\in \{x,y\}$, $p_k\in S_0$, and $k\ge 3$ (since $S_0$ is assumed to be non-empty).
	Let $p=p_2$, then $p$ is a neighbor of at least one of $x$ or $y$. Since $N(x)\neq N(y)$ and $H$ is connected we can choose a vertex $q$, such that both $x$ and $y$ have a neighbor in $\{p,q\}$, and at least one of $p$ or $q$ has exactly one neighbor in $\{x,y\}$. 
	Note that, according to~\eqref{eq:distinguishersuffices}, it does not cause any issues to rename the vertices $x$ and $y$ because a realizable $(y,x,S)$-distinguisher with respect to $(a,c)$ also gives a realizable $(x,y,S)$-forcer with respect to $(a,c)$.
	So, without loss of generality (by renaming $x$ and $y$), we can assume that $p\in N(x)$ and $q\in N(y)\setminus N(x)$.
	
%	Let $S'_0$ be a minimal-size set with the properties $p_{k-1}\in S'_0$ and $S_0\subseteq N(S'_0)$. Thus, $|S'_0|\le |S_0|$. We set $S'=\{p,q\}\cup S'_0$.
	Let $S'_0$ be a minimal-size set with the properties $p_{k-1}\in S'_0$ and $S_0\subseteq N(S'_0)$. Thus, $|S'_0|\le |S_0|$.
	Let $S'=\{p,q\}\cup S'_0$.
	Clearly, $S'$ is a one-sided set with $|S'|\le |S|$. 
	
	\begin{description}
		\item[Case~1] Suppose that $k=3$ and consequently $p=p_2=p_{k-1}$. Then $|S'|<|S|$. So, for any distinct $x', y'\in S'$, we have $(x', y', S')<(x,y,S)$, and consequently, by the induction hypothesis, there is an $(x', y',S')$-forcer with respect to $(a,c)$ that is realizable by $H$. From Corollary~\ref{cor:partitioners} we obtain that for every partition $(X,Y)$ of $S'$, the $(X,Y)$-partitioner with respect to $(a,c)$ is realizable by $H$. In particular, let $R'$ be the $(S'\cap N(x),S'\setminus N(x))$-partitioner. Then $R'$ is realizable by $H$, and $R=R_{S\to S'}\con R'$ is also realizable by $H$ (Lemma~\ref{lem:simplecombinations}). We now show that $R$ is an $(x,y,S)$-distinguisher with respect to $(a,c)$. We have $p\in R_{S\to S'}(x)= S'\cap N(x)$ and, for each $v\in S'\cap N(x)$, we have $R'(v)=\{a\}$. It follows that $R(x)=\{a\}$. Since $q\in R_{S\to S'}(y)$ but $q\notin N(x)$, we have $R'(q)=\{c\}$ and consequently $c\in R(y)$. Since $S_0\subseteq N(S'_0)$ every vertex in $S$ has a neighbor in $S'$. Therefore, for $u\in S$, at least one of $(u,a)$ or $(u,c)$ is in $R$. This shows that $R$ is an $(x,y,S)$-distinguisher with respect to $(a,c)$. By~\eqref{eq:distinguishersuffices}, we are done.
		\item[Case~2] Suppose $k>3$, i.e.,~$k\ge 5$ (since $\{x,y\}$ is one-sided). We know that $|S'_0|\le |S_0|$ and consequently $|S'|\le |S|$ --- but now these sets could have the same cardinality. However, by the choice of $P$, note that $\dist(\{p,q\}, S'\setminus\{p,q\})\le k-2 <k =\dist(\{x,y\}, S\setminus\{x,y\})$. So, we have $(p,q,S')<(x,y,S)$ and by the induction hypothesis there is a $(p,q,S')$-forcer $R'$ with respect to $(a,c)$.
		By the choice of $P$ it also follows that no vertex in $S_0$ has a common neighbor with $x$ or $y$, and therefore $x$ has no neighbor in $S'_0$. Thus, $p$ is the only neighbor of $x$ in $S'$. Consequently, $R=R_{S\to S'}\con R'$ is an $(x,y,S)$-distinguisher with respect to $(a,c)$, and this is what we need according to~\eqref{eq:distinguishersuffices}.
	\end{description}
	This completes the proof.
\end{proof}

We can now prove Lemma~\ref{lem:allrelations}, which we restate for convenience.
\allrelations*
\begin{proof}
	Let $S'$ be a maximal-size irredundant superset of $S$ in $V$. Note that every maximal-size irredundant set in $H$ contains a vertex from each class of vertices with identical neighborhood. Thus, $H'=H[S']$ is connected as $H$ is connected.
Furthermore, we have $\irr(H')=\irr(H)\ge 2$. Since $H'$ is an induced subgraph of $H$ it suffices to show that every relation in $S^p$ is realizable by $H'$.
	Since $\irr(H')\ge 2$, the graph $H'$ is not a complete bipartite graph and therefore it contains an induced four-vertex path $(a - b - c - d)$.
	By Lemma~\ref{lem:Distinguisher}, for every pair of distinct $x,y\in S$, there is an $(x,y,S)$-forcer $R$ with respect to $(a,c)$ such that $R$ is realizable by $H'$.
	The statement of the lemma follows from Lemma~\ref{lem:allrelationsSac} for $q=0$.
\end{proof}

\section{Counting list homomorphisms to general graphs $H$} \label{sec:general}
In this section we discuss how to lift the results from Section~\ref{sec:bipartite}, where we assumed $H$ to be bipartite, to the general case.

\subsection{Associated bipartite graphs}
For a graph $H=(V,E)$, by $H^*$ we denote its \emph{associated bipartite graph}, i.e.,
the graph with vertex set $\{v',v'' ~|~ v \in V\}$ and edge set $\{u'v'', u''v' ~|~ uv \in E\}$.
We set $V' := \{v' ~|~ v \in V\}$ and $V'' := \{v'' ~|~ v \in V\}$. So $(V', V'')$ is a bipartition of $H^*$.
Note that if $H = H_1 + H_2$, where $+$ denotes the disjoint sum, then $H^* = H_1^* + H_2^*$.

Recall that if $H$ is connected and nonbipartite, then $\irr(H)$ is the cardinality of the largest irredundant set in $H$.
Let us point out that associated bipartite graphs allow us to provide a uniform definition of $\irr(H)$, which does not need to distinguish bipartite and nonbipartite graphs.
\begin{obs}
For every graph $H$ it holds that $\irr(H) = \irr(H^*)$.
\end{obs}
\begin{proof}
First, observe that $S \subseteq V(H)$ is irredundant if and only if $S' = \{v' ~|~ v \in S\}$ is irredundant in $H^*$ if and only if $S'' = \{v'' ~|~ v \in S\}$ is irredundant in $H^*$.

Furthermore, we observe that there is a correspondence between connected components of $H$ and connected components of $H^*$.
Let $H'$ be a connected component of $H$.
If $H'$ is bipartite, then $H'^*$ consists of two disjoint copies of $H'$ and so $\irr(H') = \irr(H'^*)$.
If $H'$ is nonbipartite, then $H'^*$ is connected (and bipartite). Indeed, this follows from the fact that for every two vertices $x,y$ of $H'$ there is an even $x$-$y$-walk and an odd $x$-$y$-walk in $H'$. Thus all vertices $x',x'',y',y''$ are in the same connected component of $H'$. 

Consequently, each bipartite component $H'$ of $H$ corresponds to two components of $H^*$, both isomorphic to $H'$,
and each nonbipartite component $H'$ of $H$ corresponds to one connected component of $H^*$, i.e., $(H')^*$.

Now the claim easily follows from the previous observations about $S$, $S'$ and $S''$.
\end{proof}

Note that $H^*$ is a biclique if and only if $H$ is a reflexive clique.
Thus, we observe that $\irr(H) \geq 2$ if and only if $H$ has a connected component that is not a biclique nor a reflexive clique.
This allows us to restate the complexity dichotomy for $\LHom{H}$, which was originally observed as a simple consequence of the non-list result of Dyer and Greenhill~\cite{DyerG00:RSA} by D{\'{\i}}az, Serna, and Thilikos~\cite{DBLP:conf/dimacs/DiazST01}, and also by Hell and Ne\v{s}et\v{r}il~\cite{DBLP:conf/dimacs/HellN01}.

\begin{thm}[\cite{DyerG00:RSA,DBLP:conf/dimacs/DiazST01,DBLP:conf/dimacs/HellN01}]
Let $H$ be a fixed graph.
If $\irr(H)=1$, then $\LHom{H}$ is polynomial-time solvable, and otherwise it is \sharpP-complete.
\end{thm}

\subsection{Algorithm for general graphs $H$}

For an instance $(G,L)$ of $\LHom{H}$, we define its \emph{associated instance} $(G^*,L^*)$ of $\LHom{H^*}$, where for all $v\in V(G)$ 
\[
L^*(v') = \{x' ~|~ x \in L(v)\} \text{ and } L^*(v'') = \{x'' ~|~ x \in L(v)\}.
\]

We say that a homomorphism $f \from G^* \to H^*$ is \emph{clean} if, for every $v \in V(G)$ and $x \in V(H)$ it holds that $f(v') = x'$ if and only if $f(v'') = x''$.

\begin{lem} \label{lem:clean}
There is a bijection between homomorphisms from $(G,L)$ to $H$ and clean homomorphisms from $(G^*,L^*)$ to $H^*$.
\end{lem}
\begin{proof}
Consider a homomorphism $f \from (G,L) \to H$.
Let $\sigma(f)$ be a mapping $f^* \from V(G^*) \to V(H^*)$ defined as follows.
Let $v \in V(G)$ and suppose $f(v) = x$.
We set $f^*(v') = x'$ and $f^*(v'')=x''$.
It is straightforward to verify that $f^*$ is a clean homomorphism from $(G^*,L^*)$ to $H^*$ and that $\sigma$ is a bijection.
\end{proof}

Using Lemma~\ref{lem:clean} we can show the algorithmic statement of Theorem~\ref{thm:mainthm}. 
Note that it follows from the subsequent slightly stronger result.

\begin{thm}\label{thm:algogeneral}
For each graph $H$, the $\LHom{H}$ problem on $n$-vertex instances given along with a tree decomposition of width at most $t$
can be solved in time $\irr(H)^t \cdot n^{\bigO(1)}$.
\end{thm}
\begin{proof}
Consider an instance $(G,L)$ of $\LHom{H}$.
By Lemma~\ref{lem:clean}, it suffices to count the number of clean homomorphisms from $(G^*,L^*)$ to $H^*$.

Let $\calT$ be a tree decomposition of $G$ with width at most $t$. We modify it as follows: in every bag of $\calT$ we replace every vertex $v \in V(G)$ with two vertices $v',v'' \in V(G^*)$. It is straightforward to verify that this way we obtain a tree decomposition of $G^*$ with width at most $2t$; let us call this decomposition $\calT^*$.

First, just like in the proof of Theorem~\ref{thm:bipartitealgo}, we reduce the problem to its equivalent weighted version with each list of size at most $\irr(H^*) = \irr(H)$.
Consider one bag of $\calT^*$ and recall that we are only interested in counting the number of clean homomorphisms from $G^*$ to $H^*$.
Therefore, even though the size of the bag is at most $2t$, there are at most $\irr(H)^t$ colorings that could possibly be extended to a clean homomorphism from $G^*$ to $H^*$.

Thus, by an argument analogous to the one in the proof of Theorem~\ref{thm:bipartitealgo}, we obtain the desired running time.
\end{proof}

\subsection{Hardness for general graphs $H$}

Recall the definition of consistent instances from the beginning of Section~\ref{sec:bipartite}.
The following lemma is a crucial tool used in our hardness reduction.

\begin{lem} \label{lem:biptogeneral}
For a consistent instance $(G,L')$ of $\LHom{H^*}$, define $L \from V(G) \to 2^{V(H)}$ as follows: $L(v) := \{ x ~|~ \{x',x''\} \cap L'(v) \neq \emptyset \}$.
There is a bijection between homomorphisms from $(G,L')$ to $H^*$ and homomorphisms from $(G,L)$ to $H$.
\end{lem}
\begin{proof}
For a homomorphism $f \from (G,L') \to H^*$, define $\sigma(f) \from V(G) \to V(H)$ as follows.
If $f(v) \in \{x',x''\}$, then $\sigma(f)(v) = x$.
It is straightforward to verify that $\sigma(f)$ is a homomorphism from $(G,L)$ to $H$.
Furthermore, $\sigma$ is a bijection -- as the instance $(G,L')$ is consistent, for no $v \in V(G)$ and $x \in V(H)$ it holds that $x',x'' \in L'(v)$.
\end{proof}

Now we are ready to prove the complexity statement in Theorem~\ref{thm:mainthm}.
Again, we will actually show a slightly stronger result, where we consider \emph{pathwidth} as the parameter.

\begin{thm}\label{thm:hardgeneral}
Let $H$ be a graph with $\irr(H)\ge 2$.
Assuming the \#SETH, there is no $\epsilon >0$, such that $\LHom{H}$ on connected bipartite $n$-vertex instances $G$ can be solved in time $(\irr(H)-\epsilon)^{\pw(G)} \cdot n^{\bigO(1)}$, even if $\pw(G)=\tw(G)$ and $G$ is given along with an optimal path decomposition.
\end{thm}
\begin{proof}
Note that the associated graph $H^*$ contains a connected component other than a biclique and therefore $\irr(H^*)\ge 2$.
We reduce from $\LHom{H^*}$. Let $(G',L')$ be an instance of $\LHom{H^*}$, where $G'$ has $n$ vertices and is given along with a path decomposition $\calP$ with width $t$. Recall from the proof of Theorem~\ref{thm:bipartitealgo} that we can assume that the instance $(G',L')$ is consistent.

Let $(G',L)$ be an instance of $\LHom{H}$ constructed as in Lemma~\ref{lem:biptogeneral}, note that the lists $L$ are computed in polynomial time.
By Lemma~\ref{lem:biptogeneral}, we know that $\hom{(G',L')}{H^*} = \hom{(G',L)}{H}$.
Now we modify $G'$ to make sure that the treewidth and the pathwidth of our instance is \emph{exactly} $t$.

Let $v$ be an arbitrary vertex of $G'$ that appears in the last bag of $\calP$.
Let $G$ be the graph obtained from $G'$ by adding $2t-1$ new vertices, which, together with $v$, induce a biclique $K_{t,t}$ and are nonadjacent to any other vertices of $G'$.
Let $A$ and $B$ be the bipartition classes of this biclique, where $v \in A$.

The number of vertices of $G$ is $n+2t-1 \leq 3n$.
Furthermore, as $\tw(G') \leq \pw(G') \leq t$ and $\tw(K_{t,t})=\pw(K_{t,t})=t$, we conclude that $\tw(G) = \pw(G) = t$.
Finally, it is easy to obtain an optimal path decomposition of $G$ by appending to $\calP$ an optimal path decomposition of $K_{t,t}$, where $v$ appears in the first bag.

For each $a \in L(v)$, let $a'$ be an arbitrary neighbor of $a$ in $H$.
Let $L_a : V(G) \to 2^{V(H)}$ be the list function defined as follows:
\[
L_a(u) = \begin{cases}
\{a\}		& \text{ if } u \in A,\\
\{a'\}		& \text{ if } u \in B,\\
L(u)	& \text{ otherwise.}\\
\end{cases}
\]

Since lists $L_a$ force a unique coloring of vertices in $A \cup B$, we observe that $\hom{(G,L_a)}{H}$ is equal the number of homomorphisms from $(G',L)$ to $H$ that map $v$ to $a$.
Consequently, we obtain
\[\hom{(G',L')}{H^*} = \hom{(G',L)}{H} = \sum_{a \in L(v)}\hom{(G,L_a)}{H}.\]

Thus, if for each $a \in L(v)$ we could compute $\hom{(G,L_a)}{H}$ in time $(\irr(H)-\epsilon)^{\pw(G)} \cdot (n+2t-1)^{\bigO(1)}$,
then we could also compute $\hom{(G',L')}{H^*}$ in time
\[(\irr(H)-\epsilon)^{\pw(G)} \cdot (n+2t-1)^{\bigO(1)}=(\irr(H^*)-\epsilon)^{t} \cdot n^{\bigO(1)}.\]
Recall that $\irr(H^*)\ge 2$.
By Theorem~\ref{thm:bipartitehardness}, the existence of such an algorithm for $\LHom{H^*}$ contradicts the \#SETH.
\end{proof}

\subsection{From $\LHom{H}$ to $\Hom{H}$: special cases}
In this section we show two corollaries of our main result that concern \emph{non-list} problems.
First, let us show that Corollary~\ref{cor:hardnessp4} implies the following lower bound.

\begin{cor}\label{cor:bipartiteis}
There is no algorithm that counts all independents sets in $n$-vertex bipartite graphs given with a path decomposition of width $t$ in time $(2-\epsilon)^t \cdot n^{\bigO(1)}$ for any $\epsilon >0$, unless the \#SETH fails.
\end{cor}
\begin{proof}
We give a pathwidth-preserving reduction from $\LHom{P_4}$.
Let $(G,L)$ be an instance of $\LHom{P_4}$, where $P_4 = (a-b-c-d)$. Note that we can safely assume that $G$ is bipartite, as otherwise the answer is clearly 0.
Furthermore, similarly as in the proof of Theorem~\ref{thm:bipartitealgo}, we can assume that $G$ is connected.
Let $(X,Y)$ be a bipartition of $G$.

Let $L'$ (resp. $L''$) be the list function defined by setting $L'(x) = L(x) \cap \{a,c\}$ and $L'(y) = L(y) \cap \{b,d\}$ (resp. $L''(x) = L(x) \cap \{b,d\}$ and $L''(y) = L(y) \cap \{a,c\}$) for every $x \in X$ and $y \in Y$.
Observe that in any homomorphism $f$ from $G$ to $P_4$, either $f(X) \subseteq \{a,c\}$ and $f(Y) \subseteq \{b,d\}$,
or $f(X) \subseteq \{b,d\}$ and $f(Y) \subseteq \{a,c\}$. Thus we have
\[
\hom{(G,L)}{P_4}  = \hom{(G,L')}{P_4} + \hom{(G,L'')}{P_4}.
\]

Let us focus on computing $\hom{(G,L')}{P_4}$, as computing $\hom{(G,L'')}{P_4}$ is symmetric.
We observe that if for some vertex $v \in V(G)$ we have $L'(v) = \{a\}$ (resp. $L'(v)=\{d\}$), then we can safely remove $d$ (resp. $a$)
from the lists of all neighbors of $v$, and then delete the vertex $v$.
On the other hand, if $L'(v) =\{b\}$ or $L'(v) =\{c\}$, then we can safely remove $v$ from the graph.
Let $(\widetilde{G}, \widetilde{L})$ the the instance obtained by exhaustive application of the rules above and note that 
\[
\hom{(\widetilde{G},\widetilde{L})}{P_4}  = \hom{(G,L')}{P_4}.
\]
If the list of some vertex is empty, then $\hom{(\widetilde{G},\widetilde{L})}{P_4}=0$ and we are done.
Otherwise, every list in $\widetilde{L}$ is either $\{a,c\}$ (for vertices in $X$) or $\{b,d\}$ (for vertices in $Y$).
Consequently, $\hom{(\widetilde{G},\widetilde{L})}{P_4}$ is precisely the number of independent sets in $\widetilde{G}$ (the independent set is formed by the vertices mapped to $a$ and $d$). 

Recall that $\widetilde{G}$ is a subgraph of $G$. Thus, given a path decomposition of $G$ with width at most $t$,
we can easily obtain a path decomposition of $\widetilde{G}$ with width at most $t$.
Now the lower bound follows directly from Corollary~\ref{cor:hardnessp4}.
\end{proof}

Next, let us focus on the case that $H = K_q$ for $q \geq 3$, i.e., counting proper $q$-colorings of a given graph.
As this problem is a special case of $\LHom{K_q}$, the algorithmic statement follows immediately from Theorem~\ref{thm:algogeneral}.
However, Theorem~\ref{thm:hardgeneral} proves hardness only for counting \emph{list} $q$-colorings.
Let us show a reduction from the problem of counting list $q$-colorings to the problem of counting $q$-colorings.
For simplicity, we will only prove hardness parameterized by the treewidth of the instance graph, but using the same approach as in Theorem~\ref{thm:hardgeneral} we can obtain an analogous result parameterized by the pathwidth. 
Let us point out that the lower bound for counting proper $q$-colorings, conditioned on the SETH, follows from the result for the decision variant of the problem by Lokshtanov, Marx, Saurabh~\cite{DBLP:journals/talg/LokshtanovMS18}.

\begin{cor}\label{cor:coloring}
Let $q \geq 3$. On $n$-vertex instances with treewidth $\tw$, the number of proper $q$-colorings cannot be counted in time $(q-\epsilon)^\tw \cdot n^{\bigO(1)}$, for any $\epsilon >0$, even if a tree decomposition of width $\tw$ is given as part of the input, unless the \#SETH fails.
\end{cor}
\begin{proof}
Let $G = (V,E)$ be an $n$-vertex graph and let $L \from V \to 2^{[q]}$ be a list function.
Suppose that the treewidth of $G$ is $\tw$ and $G$ is given along with a tree decomposition $\calT$ of width $\tw$.

Let $G'$ be a graph obtained from $G$ as follows. First, we introduce a $q$-vertex clique $K$ with vertices $\{x_1,x_2,\ldots,x_q\}$.
Next, for each $v \in V$ and each $i \in [q]$, we make $v$ adjacent to $x_i$ if and only if $i \notin L(v)$.

It is straightforward to observe that $(G,L)$ is a yes-instance of list $q$-coloring if and only if $G'$ is a yes-instance of $q$-coloring.
Furthermore, each proper list coloring of $(G,L)$ corresponds to exactly $q!$ proper colorings of $G'$, one for each proper coloring of vertices of $K$.
Thus $\hom{(G,L)}{K_q} = \frac{1}{q!} \cdot \hom{G'}{K_q}$.

The number of vertices of $G'$ is $n+q$. Now let us deal with the treewidth.
We can easily modify the tree decomposition $\calT$ of $G$ into a tree decomposition $\calT'$ of $G'$ by including all vertices of $K$ in every bag.
This proves that the treewidth of $G'$ is at most $\tw + q$. Let us further modify the instance, so that the treewidth is exactly $\tw + q$. We use a trick similar to the one in the proof of Theorem~\ref{thm:hardgeneral}.

Let $v$ be an arbitrary vertex of $G'$ and let $G''$ be obtained from $G'$ by introducing $2(\tw + q)-1$ new vertices, which, together with $v$, form a biclique $K_{\tw+q,\tw+q}$. Recall that the treewidth of $K_{\tw+q,\tw+q}$ is exactly $\tw+q$, and the treewidth of $G'$ is also $\tw + q$. On the other hand we can turn the tree decomposition $\calT'$ of $G'$ into a tree decomposition $\calT''$ of $G''$ as follows.
Let $\widetilde \calT$ be an optimal tree decomposition of the biclique $K_{\tw+q,\tw+q}$.
We choose any bag of $\calT'$ containing $v$ and make it adjacent to any bag of $\widetilde \calT$ containing $v$. This way we obtain a tree decomposition of $G''$ with width $\tw + q$.
The number of vertices of $G''$ is $n + q + 2(\tw + q)-1 = \bigO(n)$.

Now, let $f(q,\tw)$ be the number of proper $q$-colorings of $K_{\tw+q,\tw+q}$ with the color of one vertex fixed.
We observe that $\hom{G'}{K_q} = \frac{1}{f(q,\tw)} \cdot \hom{G''}{K_q}$.

Thus, if we could count the number of proper $q$-colorings of $G''$ in time $(q - \epsilon)^{\tw + q} n^{\bigO(1)}$, we could count the number of list $q$-colorings of $(G,L)$ in time  $(q - \epsilon)^{\tw + q} n^{\bigO(1)} = (q - \epsilon)^{\tw} n^{\bigO(1)}$. By Theorem~\ref{thm:algogeneral}, this would contradict the \#SETH.
\end{proof}

\section{Conclusion}
Let us conclude the paper with the discussion of a potential extension of our results to the non-list variant of the $\LHom{H}$ problem, i.e., the problem of counting homomorphisms to a fixed graph $H$. Denote this problem by $\Hom{H}$.
The complexity dichotomy for $\Hom{H}$ was provided by Dyer and Greenhill~\cite{DyerG00:RSA} and it is exactly the same as for the list variant: the problem is polynomial-time solvable if every component of $H$ is either a reflexive clique or an irreflexive biclique, and otherwise it is $\#$P-complete.

While the algorithmic statement of Theorem~\ref{thm:mainthm} clearly carries over to  $\Hom{H}$ (as  $\Hom{H}$ is a restriction of $\LHom{H}$ where all lists are equal to $V(H)$), our hardness proof heavily exploits non-trivial lists.
The simple tricks we used in Corollaries~\ref{cor:bipartiteis} and~\ref{cor:coloring} to reduce the list variant of coloring to the non-list variant cannot be easily generalized to arbitrary graphs $H$.

Let us point out that the fine-grained complexity of the decision variant of $\Hom{H}$, parameterized by the treewidth of the instance graph is not fully understood~\cite{DBLP:journals/siamcomp/OkrasaR21}.

Typically, the hardness proofs concerning the complexity of (non-list) graph homomorphism problems involve some tools from universal algebra and algebraic graph theory~\cite{DBLP:journals/jct/HellN90,DBLP:journals/tcs/Bulatov05,Siggers}. For graphs $H_1,H_2,\ldots,H_p$, we define their \emph{direct product} $H_1 \times \cdots \times H_p$ as follows:
\begin{align*}
V(H_1 \times \cdots \times H_p) = & V(H_1) \times V(H_2) \times \cdots \times V(H_p),\\
E(H_1 \times \cdots \times H_p) = & \{ (x_1,\ldots,x_p)(y_1,\ldots,y_p) ~|~ x_iy_i \in E(H_i) \text{ for all } i \in [p] \}.
\end{align*}

The following observation is straightforward.

\begin{obs}
Let $G$ and $H = H_1 \times \cdots \times H_p$ be two graphs.
A function $f \from V(G) \to V(H)$ is a homomorphism if and only if for every $i \in [p]$, the function $\Pi_i(f)$ is a homomorphism from $G$ to $H_i$.
\end{obs}
As an immediate consequence, we obtain 
\[
\hom{G}{H_1 \times \cdots \times H_p} = \prod_{i =1}^p \hom{G}{H_i}.
\]
In other words, if $H$ can be obtained as a direct products of some factors $H_1,H_2,\ldots,H_p$, we can reduce solving $\Hom{H}$ to solving $\Hom{H_i}$ for all $i$ (with the same instance graph).
As $\irr(H_i)$ can be much smaller than $\irr(H)$, we observe that for some graphs $H$, the $\Hom{H}$ problem can be solved faster than the $\LHom{H}$ problem.
In particular, while it was the truth for $\LHom{H}$, the parameter $\irr(H)$ is not always the correct base of the exponential factor appearing in the complexity of an optimal algorithm solving $\Hom{H}$.

\bibliography{CountLHomTW}
\end{document}